\documentclass[journal,11pt,onecolumn]{IEEEtran}
\usepackage{lipsum}

\usepackage[utf8]{inputenc}
\usepackage{setspace}% provides font selecting commands
\usepackage{hyperref}% provides font selecting commands
\usepackage[table]{xcolor}
\usepackage{enumitem}

\usepackage{amsthm}
\usepackage{array}
\usepackage{graphicx}
\usepackage{amsfonts}
\usepackage{multicol}
\usepackage[cmex10]{amsmath}
\usepackage{mathtools}
\usepackage[amsmath]{empheq}
\usepackage{cases}
\usepackage{balance}
\usepackage{cite}
\usepackage{lineno}
\usepackage{booktabs}
\usepackage{mathrsfs}
\usepackage{bbm}
\usepackage{amssymb}
\usepackage{soul}

\usepackage[paperwidth=9in, paperheight=11in]{geometry}

\usepackage[T1]{fontenc}

\definecolor{myblue}{rgb}{0,0,1}

\let\OldLambda\lambda
\renewcommand\lambda{\boldsymbol{\OldLambda{}}}

\def\l{\left(}
\def\r{\right)}
\def\K{\mathcal{K}}

\def\j{\jmath}
\def\DE{\stackrel{\mathrm{def}}{=}}

\newtheorem{theorem}{Theorem}
\newtheorem{definition}{Definition}

\newtheorem{proposition}{Proposition}

\newcommand{\EQc}[1]{\stackrel{(\ref{#1})}{=}}
\def\ind{\mathbbmtt{1}}
\def\bsl{\boldsymbol{\Lambda}}
\def\vbl{{\sf{V}}_{\sf{BL}}}
\def\vfs{{\sf{V}}_{\sf{FS}}}
\def\bsls{ \boldsymbol{\Lambda}_{\sf S} }
\def\bsll{ \boldsymbol{\Lambda}_{\sf L} }
\def\bslsi{\boldsymbol{\Lambda}_\mathsf{S}^{\textsf{inv}}}
\def\PH{\Phi_{\bsls}}
\def\PHI{\Phi_{\bslsi}}

\def\iS{{\mathscr{T}_{\bslsi}}}
\def\SC{*_{\bsls}}
\def\SCI{*_{\bslsi}}

\def\SL{{\sf SL}_2\l \mathbb{R}\r}

\def\K{{\color{black}K_{\bsls}}}
\def\Ki{{\color{black}K^*_{\bsls}}}
\def\CB{{\color{black}K_{\bslsi}}}

\def\j{{\jmath}}
\def\DE{\stackrel{\mathrm{def}}{=}}

\newcommand{\RB}[2]{#1\l #2 \r}

\newcommand{\FT}[1]{ \mathscr{T}_{ \boldsymbol{\Lambda}_{{\mathsf {FT}}} } \left[ #1 \right]}
\newcommand{\SAFT}[1]{ \mathscr{T}_{ \boldsymbol{\Lambda}_{\sf S} } \left[ #1 \right]}

\newcommand{\FTN}[2]{ \mathscr{T}^{\l  #2 \r}_{ \boldsymbol{\Lambda}_{{\mathsf {FT}}} } \left[ #1 \right]}
\newcommand{\saft}[1]{\widehat{#1}_{\boldsymbol{\Lambda}_\mathsf{S}} \left( \omega \right) }

\newcommand{\DEq}[1]{\stackrel{(\ref{#1})}{=}}
\newcommand{\DA}[1]{\stackrel{(\mathrm{#1})}{=}}

\newcommand{\nvec}[1]{ {{\mathbf{#1}}}}

\newcommand{\saftL}[2]{ \mathscr{T}
_{ \boldsymbol{\Lambda}_{{\mathsf S}_{#2}} } \left[ #1 \right]}

\newcommand{\saftC}[1]{ \mathscr{T}_{ \boldsymbol{\Lambda}_{{\mathsf S}_{#1}} }}

\newcommand{\upo}[1]{ {\overset{\lower0.5em\hbox{$\smash{\scriptscriptstyle\rightharpoonup}$}} {{#1}}}}
\newcommand{\dno}[1]{ {\overset{\lower0.5em\hbox{$\smash{\scriptscriptstyle  \leftharpoonup}$}} {{#1}}}}
\newcommand{\up}[2]{ {\overset{\lower0.5em\hbox{$\smash{\scriptscriptstyle\rightharpoonup}$}} {{#1}}} \left( {#2} \right)}
\newcommand{\dn}[2]{ {\overset{\lower0.5em\hbox{$\smash{\scriptscriptstyle  \leftharpoonup}$}} {{#1}}} \left( {#2} \right)}
\newcommand{\GT}[4]{\widehat{S}_{\boldsymbol{\Lambda}_\mathsf{S}}^{{#1},{#2}}\left( {{#3} ,{#4}} \right)}
\newcommand{\sincD}[1]{{\operatorname{sinc} _\Delta }\left( #1 \right)}
\newcommand{\app}[1]{{\breve{#1}}}
\newcommand{\etal}{\textit{et al}.~}

\newcounter{chem}
\newcounter{temp}

\usepackage{mhchem}

\newenvironment{chequation}{%
  \setcounter{temp}{\value{equation}}%
  \setcounter{equation}{\value{chem}}%
}{%
  \setcounter{chem}{\value{equation}}%
  \setcounter{equation}{\value{temp}}%
}

\ifCLASSINFOpdf
 \else
\fi
\begin{document}

\title{Sampling and Super-resolution of Sparse Signals Beyond the Fourier Domain}

\author{Ayush~Bhandari and~Yonina~C.~Eldar
\thanks{Preliminary ideas leading to this manuscript were presented in parts at ICASSP 2015 \cite{Bhandari:2015} and 2016 \cite{Bhandari:2016}.}
\thanks{A. Bhandari is with the Massachusetts Institute of Technology, Cambridge, MA 02139-4307 USA. e-mail: \textrm{ayush@MIT.edu}.}% <-this % stops a space
\thanks{Y. Eldar is with the Technion--Israel Institute of Technology, Haifa 32000, Israel. e-mail: \textrm{yonina@ee.technion.ac.il}.}% <-this % stops a space
\thanks{The work of Y. Eldar was funded by the European Union's Horizon 2020 research and innovation program under grant agreement ERC-BNYQ, by the Israel Science Foundation under Grant no. 335/14, and by ICore: the Israeli Excellence Center `Circle of Light'.}
\thanks{Manuscript received Month XX, 20XX; revised Month XX, 20XX.}}

% The paper headers
\markboth{Journal of \LaTeX\ Class Files,~Vol.~Xx, No.~X, Month~20XX}%
{Authors \MakeLowercase{\textit{et al.}}: IEEEtran.cls for IEEE Journals}

% make the title area
\maketitle

\begin{abstract}
Recovering a sparse signal from its low-pass projections in the Fourier domain is a problem of broad interest in science and engineering {and is commonly referred to as super-resolution}. In many cases, however, Fourier domain may not be the natural choice. For example, in holography, {low-pass projections of sparse signals are obtained in the Fresnel domain}. Similarly, time-varying system identification relies on {low-pass projections on the space of linear frequency modulated signals.} In this paper, we study the recovery of sparse signals from low-pass projections in the Special Affine Fourier Transform domain (SAFT). The SAFT parametrically generalizes a number of well known unitary transformations that are used in signal processing and optics. In analogy to the Shannon's sampling framework, we specify sampling theorems for recovery {of} sparse signals {considering} three specific cases: (1) sampling with arbitrary, bandlimited kernels, (2) sampling with smooth, time-limited kernels and, (3) recovery from Gabor transform measurements linked with the SAFT domain. Our work offers a unifying perspective on the sparse sampling problem {which is} compatible with the Fourier, Fresnel and Fractional Fourier domain based results. {In deriving our results}, we introduce the SAFT series (analogous {to} the Fourier series) and the short time SAFT, and study convolution theorems that establish a convolution--multiplication property in the SAFT {domain}.
%Within the field of signal processing, some variations on this theme include, time-delay estimation, sparse deconvolution, super-resolution, multi-path estimation and sparse/FRI sampling. Despite the prevalence of the sparsity assumption, in many cases, Fourier domain may not be the natural choice. For example, in holography, sparsity is associated with the Fresnel transform. Similarly, time-varying system identification relies on sparse sum of linear frequency modulated signals. Consequently, a compelling requirement is to go beyond the classical assumption of Fourier spectrum and consider something more general. This is the theme of our work. 
\end{abstract}

% Note that keywords are not normally used for peerreview papers.
%\begin{IEEEkeywords}
%Finite rate of innovation, fractional Fourier domain, sampling, sparsity, special affine Fourier transform, {super-resolution}.
%\end{IEEEkeywords}

\newpage
\tableofcontents
\newpage

\doublespacing
% As a general rule, do not put math, special symbols or citations
% in the abstract or keywords.

% For peer review papers, you can put extra information on the cover
% page as needed:
% \ifCLASSOPTIONpeerreview
% \begin{center} \bfseries EDICS Category: 3-BBND \end{center}
% \fi
%
% For peerreview papers, this IEEEtran command inserts a page break and
% creates the second title. It will be ignored for other modes.
%\IEEEpeerreviewmaketitle

\section{Introduction}
\label{sec:int}

% The very first letter is a 2 line initial drop letter followed
% by the rest of the first word in caps.
% 
% form to use if the first word consists of a single letter:
% \IEEEPARstart{A}{demo} file is ....
% 
% form to use if you need the single drop letter followed by
% normal text (unknown if ever used by the IEEE):
% \IEEEPARstart{A}{}demo file is ....
% 
% Some journals put the first two words in caps:
% \IEEEPARstart{T}{his demo} file is ....
% 
% Here we have the typical use of a "T" for an initial drop letter
% and "HIS" in caps to complete the first word.
\IEEEPARstart{T}{he} problem of super-resolution deals with recovery {of spikes or Dirac masses from low-pass projections in the Fourier domain. This is a standard problem with numerous applications in science and engineering}. In this setting, the measurements amount to a stream of smooth pulses where the low-pass nature is due to the pulse shape. {This may be an excitation pulse used in time-of-flight imaging such as radar, sonar, lidar and ultrasound. The pulse shape may also be represented as a low-pass filter that is an approximation of,}

\begin{enumerate}[leftmargin=15pt,label=$\bullet$,itemsep = 0pt]
\item point spread function of an optical instrument such as a lens, microscope or a telescope.
\item transfer function of an electronic sensor such as an antenna or a microphone.
\item Green's function of some partial differential equation that represents a physical process (for example, diffusion or fluorescence lifetime imaging).
\item beampattern of a sensor array. % \cite{Fuchs:1994}. 

\item spectral line shape in spectroscopy (often assumed to be a Cauchy, Gaussian or a Voigt distribution). 
\end{enumerate} 

Several well known signal processing {applications involve super-resolution \cite{Donoho:1992,Manabe:1992,Candes:2013} including source localization\cite{Fuchs:1994}, time-delay estimation \cite{Zi-qiang:1982,Pallas:1991,Li:1998a,Gedalyahu:2010a}, sparse deconvolution \cite{Li:2000} and, time-of-flight imaging (e.g.~optical \cite{Hernandez-Marin:2007,Bhandari:2014,Bhandari:2016a}, radar \cite{Bar-Ilan:2014,Rudresh:2017,Bhandari:2017a} and ultrasound \cite{Tur:2011,Burshtein:2016}).} {These applications all address the same challenge:} \emph{``How can one recover a signal with broadband features (spikes) from a given set of narrowband measurements?''}

Our problem can be restated as that of {uniform} sampling and recovery of spikes with a given bandlimited kernel. Unlike bandlimited or smooth signals which follow a linear recovery principle \cite{Unser:2000,Eldar:2015}, sparse signals rely on non-linear recovery {method}. Despite the prevalence of the spike recovery problem across several fields (cf.~Table~I and Table~II in \cite{Pallas:1991}, \cite{Li:1998a} as well as \cite{Li:2000}), the link to sampling theory was established only recently by Vetterli \cite{Vetterli:2002}, Blu\cite{Blu:2008} and co-workers in their study of \emph{finite rate of innovation} or FRI signals. These are signals which are described by countable degrees of freedom, {per unit time and model a broad class of signals. The FRI sampling has been applied to a number of interesting applications} such as channel estimation \cite{Barbotin:2012}, radar \cite{Bar-Ilan:2014,Rudresh:2017,Bhandari:2017a}, time-resolved imaging \cite{Bhandari:2016a,Bhandari:2017a,Bhandari:2016b}, {sparse recovery on a sphere\cite{Deslauriers-Gauthier:2013}}, image feature detection \cite{Chen:2012,Pan:2014,Mulleti:2016}, medical imaging \cite{Onativia:2013,Tur:2011,Burshtein:2016,Dogan:2014}, tomography \cite{Seelamantula:2014}, astronomy \cite{Pan:2017}, spectroscopy \cite{Mulleti:2017}, unlimited sampling architecture \cite{Bhandari:2017b,Bhandari:2018a} and inverse source problems \cite{Murray-Bruce:2017,Pan:2017b}. 

{In either case, super-resolution or sampling of sparse signals, a common feature is that both problems assume a bandlimited kernel.} The choice of Fourier domain for defining bandlimitedness may be restrictive. In practice, many systems and physical phenomena are modeled as linear and time-varying/non-stationary. On the other hand, complex exponentials---that are constituent components of the Fourier transform---are the eigenfunctions of linear time-invariant systems. %As a result, the choice of Fourier domain may be sub-optimal in certain applications.
Polynomial phase models \cite{Peleg:1991,Yuan:2012,Amar:2010} that generalize complex exponentials are often used as an alternative basis for modeling time-varying systems. Such models are specified by basis functions of {the} form ${e^{\jmath \varphi \left( t \right)}}$. One notable example is that of quadratic chirps which are specified by $\varphi \left( t \right) = {a_2}{t^2} + {a_1}t + {a_2}$. Due to their wide applicability, chirp based transformations \cite{Mann:1995}, multi-scale orthonormal bases and frames \cite{Baraniuk:1993} as well as dictionary based pursuit algorithms \cite{Gribonval:2001} have been {derived} in the literature. Active imaging systems such as radar \cite{Engen:2011} and sonar \cite{Schock:2004} use chirps for probing the environment. In \cite{Martone:2001}, Martone demonstrates the {use} of polynomial phase basis functions of the fractional Fourier transform for multicarrier communication with time-frequency selective channels. Harms \etal \cite{Harms:2015a} use chirps for identification of linear time-varying systems. Besides chirps, Fresnel transforms \cite{Gori:1981} use polynomial phase representation for digital holography \cite{Chacko:2013} and diffraction. Other applications of polynomial phase functions include time-frequency representations \cite{Pei:2007a}, DOA estimation \cite{Yuan:2012}, sensor array processing \cite{Yetik:2003,Amar:2010}, ghost imaging \cite{Setala:2010}, image encryption \cite{Unnikrishnan:2000} and quantum physics \cite{Huang:2011}. 

Polynomial phase representations were {also} studied in the context of phase space and mathematical physics. This led to {the} development of unitary transformations such as the fractional Fourier transform (FrFT) \cite{Condon:1937} and the Linear Canonical Transform (LCT) \cite{Moshinsky:1971,Healy:2016}. These transformations generalize the Fourier transform {in the same way that} polynomial phase {functions} ${e^{\jmath \varphi \left( t \right)}}$ generalizes the complex exponentials, or ${e^{\jmath \omega t}}$. 

In the area of signal processing, Almeida first introduced the fractional Fourier transform (FrFT) as a tool for time-frequency representations \cite{Almeida:1994}. Following \cite{Almeida:1994}, a number of papers have extended the Shannon's sampling theorem to the FrFT domain (cf.\cite{Tao:2008} and \cite{Bhandari:2012} and references there in). In \cite{Bhandari:2012}, Bhandari and Zayed developed the shift-invariant model for the FrFT domain which was later extended in \cite{Shi:2012,Bhandari:2017}. Sampling of sparse signals in the FrFT domain was studied in \cite{Bhandari:2010}. {Interestingly, all of the aforementioned transformations and corresponding basis functions are specific cases of the Special Affine Fourier Transform (SAFT). }

%%
%%%%%%%%%%%%%%%%%%%%%%%%%%%%
%%%%%%%%%%%%%%%%%%%%%%%%%%%%
%%\begin{figure*}[!t]
%% \centering
%%\begin{equation}
%%\label{saftkernel}
%%{\kappa_{\bsls} }\left( {t,\omega } \right) = \K^*\exp \left( { - \frac{\j}{{2b}}\left( {a{t^2} + d{\omega ^2} + 2t\left( {p - \omega } \right) - 2\omega \left( {dp - bq} \right)} \right)} \right), \qquad \K = \frac{1}{{\sqrt {\jmath 2\pi b} }}\exp \left( {\jmath\frac{{{dp^2}}}{{2b}}} \right)
%%\tag{10}
%%\end{equation}
%%\hrule
%%\end{figure*}
%%%%%%%%%%%%%%%%%%%%%%%%%%%%
%%%%%%%%%%%%%%%%%%%%%%%%%%%%
%\subsection{Contributions}
%%In the spirit of ``sampling theory beyond Fourier domain,'' i

%\subsection{Contributions}
In this paper, our goal is to extend sampling theory of sparse signals beyond the Fourier domain. We do so by considering the SAFT {which} is a parametric transformation {that} subsumes a number of well known unitary transformations used in signal processing and optics. {We recently studied sampling theory of bandlimited and smooth signals in the SAFT domain in \cite{Bhandari:2017}.} By using results developed in \cite{Bhandari:2017}, {here we derive} sampling theorems for sparse signals with three distinct flavors: 
\begin{enumerate}[leftmargin=45pt,label={\arabic*)},itemsep = 0pt]
\item Sampling with arbitrary, bandlimited kernels. 
\item Sampling with smooth, time-limited kernels.
\item Sparse signal recovery from Gabor transform measurements linked with the SAFT domain.
\end{enumerate}
For this purpose, we introduce two mathematical tools:
\begin{enumerate}[leftmargin=15pt,label=$\bullet$,itemsep = 0pt]
\item the Special Affine Fourier Series (a generalization of the Fourier Series) for representing time-limited functions. 
\item the Gabor transform associated with the SAFT (a generalization of the usual Gabor transform). 
\end{enumerate}

%\subsection{Scope and Organization of this Paper}
We {begin} with the definition of the SAFT in Section~\ref{sec:SAFT}. The forward transform, its inverse as well as geometric properties are discussed in the subsections that follow. In order to develop sparse sampling theory for the SAFT, in Section~\ref{subsec:SAFTConv} we {recall convolution operators for the SAFT domain \cite{Bhandari:2018}. This allows us to establish the link between convolution and low-pass orthogonal projection operators. We then recall Shannon's sampling theorem for SAFT bandlimited functions\cite{Bhandari:2018} in Section~\ref{sec:ShannonSAFT}.} Our main results on sparse sampling theory are presented in Section~\ref{sec:main} where we discuss three cases. Unlike the Fourier basis functions, the SAFT counterparts are aperiodic. As a workaround, in Section~\ref{subsec:SAFS}, we develop the Special Affine Fourier Series (SAFS) for time-limited signals. The SAFS is then used to represent sparse signals and we conclude this work with several future directions in Section~\ref{sec:conclusion}.

%\subsection{Notation}
Throughout the paper, set of integers, reals and complex numbers is denoted by $\mathbb Z, \mathbb R$, and $\mathbb C$, respectively and $\mathbb{Z}^+$ denotes a set of positive integers. Continuous-time functions are denoted by $f\l t \r, t\in \mathbb{R}$ while $f\left[ m\right], m\in\mathbb{Z}$ are used for their discrete counterparts. We use script fonts for operators, that is, $\mathscr{O}f$. For instance, $\mathscr{P}$ denotes the projection operator and the derivative operator of order $k$ is written as $\mathscr{D}^k f = f^{\l k \r}$. Function/operator composition is denoted by $\circ$. We use boldface font for representing vectors and matrices, for example $\nvec{x}$ and $\mathbf{X}$, respectively, and {$\mathbf{X}^\top$ is the matrix transpose}. We use $\mathbf{I}$ to denote an identity matrix. A characteristic function on domain $\cal{D}$ is denoted by $\ind_\mathcal{D}$. Dirac distribution is represented by $\delta\l t \r, t\in \mathbb{R}$. All operations linked with $\delta$ are treated in terms of distributions. The Kronecker delta is represented by {$\delta\left[m\right], m\in\mathbb{Z}$}. The space of square-integrable and absolutely integrable functions is denoted by $L_2$ and $L_1$, respectively and $\left\langle {f,g} \right\rangle  = \int {f{g^*}}$ is the $L_2$ inner-product. We use $\overline{\l\cdot\r}$ to denote time-reversal.

%<><><><><><><><><><><><><><><><><><><><><><><><><><><>
 % 		PARAMETERS
 %<><><><><><><><><><><><><><><><><><><><><><><><><><><>
\begin{table}[t]
\footnotesize
\centering
%\caption{Linear Canonical Transform as a generalization of other well known transformations.}
\caption{ \textrm{SAFT, Transformations and Operations}}
\begin{tabular*}{0.665\textwidth}{p{4cm}  p{5.8cm}}%{ll}
\toprule
%\addlinespace
%\multicolumn{2}{c}{%
%$\boxed{\phi _{\pmb{\Lambda}}  \left( {t,\omega } \right) = \tfrac{1}
%{{\sqrt {-j2\pi b} }}\exp \left\{ { - \tfrac{j}
%{{2b}}\left( {\left( {at^2  + d\omega ^2 } \right) - 2\omega t} \right)} \right\} \qquad (1)}$} \\  
%\addlinespace
%(r){1-2}
%
\rowcolor{yellow!40} 
\hline
{ SAFT Parameters} $\left(\bsls \right) $ & { Corresponding Transform} \\
\hline
\addlinespace

$\bigl[ \begin{smallmatrix} &0 &1& \vline & & {0}  \\-&1&0 & \vline & & {0} \end{smallmatrix} \bigr] 
= \pmb\Lambda_\textsf{FT}$  													&  \textbf{Fourier Transform (FT)}   \\ 		[2.5pt] 

$\bigl[ \begin{smallmatrix} &0 &1& \vline & & {p}  \\-&1&0 & \vline & & {q} \end{smallmatrix} \bigr] 
= \pmb\Lambda_\textsf{FT}^\mathsf{O}$  	&\textbf{Offset Fourier Transform}  \\		[2.5pt] 

$\bigl[ \begin{smallmatrix} &\cos\theta&\sin\theta  & \vline & & {0}   \\
 -&\sin\theta&\cos\theta & \vline& &{0} \end{smallmatrix} \bigr] = \pmb\Lambda_\theta $				&  \textbf{Fractional Fourier Transform (FrFT)} \\	[2.5pt] 
%      			& each      \\

$\bigl[ \begin{smallmatrix} &\cos\theta&\sin\theta  & \vline & & {p}   \\
 -&\sin\theta&\cos\theta & \vline& &{q} \end{smallmatrix} \bigr] = \pmb\Lambda_\theta^\mathsf{O} $				&  \textbf{Offset Fractional Fourier Transform} \\	[2.5pt]

$\bigl[ \begin{smallmatrix} a &b  & \vline & & {0}   \\
 c& d & \vline& &{0} \end{smallmatrix} \bigr] = \pmb\Lambda_\textsf{L} $				&  \textbf{Linear Canonical Transform (LCT)} \\	[2.5pt] 

$\bigl[ \begin{smallmatrix} 1 & b & \vline & & {0}  \\0&1 & \vline & & {0} \end{smallmatrix} \bigr] = \pmb\Lambda_\textsf{Fr}$  						& \textbf{Fresnel Transform}  \\	[2.5pt] 

$\bigl[ \begin{smallmatrix} 0 & \j & \vline & & {0} \\ \j & 0 & \vline & & {0} \end{smallmatrix} \bigr] 
= \pmb\Lambda_\textsf{LT}$ & \textbf{Laplace Transform (LT)}  \\	[2.5pt] 
$\bigl[ \begin{smallmatrix} \j \cos\theta &\j \sin\theta & \vline & & {0}  \\
 \j \sin \theta & -\j\cos\theta & \vline & & {0} \end{smallmatrix} \bigr] $  								& \textbf{Fractional Laplace Transform}   \\ 	[2.5pt] 
$\bigl[ \begin{smallmatrix} 1&\jmath b & \vline & & {0}  \\  \jmath&1 & \vline & & {0} \end{smallmatrix} \bigr]$				& \textbf{Bilateral Laplace Transform}  \\	[2.5pt] 
$\bigl[ \begin{smallmatrix} 1&-\jmath b & \vline & & {0}  \\ 0&1 & \vline & & {0} \end{smallmatrix} \bigr]$, $b \ge 0$  		& \textbf{Gauss--Weierstrass Transform}   \\	[2.5pt] 
$\tfrac{1}{{\sqrt 2 }} \bigl[\begin{smallmatrix} 0 & e^{ - {{\jmath\pi } 
\mathord{\left/{\vphantom {{j\pi } 2}} \right.\kern-\nulldelimiterspace} 2}} & \vline & & {0} 
\\-e^{ - {{\jmath\pi } \mathord{\left/{\vphantom {{j\pi } 2}} \right.\kern-\nulldelimiterspace} 2}} 
&1 & \vline & & {0} \end{smallmatrix} \bigr]$  													& \textbf{Bargmann Transform} \\	[3pt] 
%
%\midrule
\addlinespace
\hline
\rowcolor{yellow!40} 
{SAFT Parameters} $\left(\bsls \right) $ & { Corresponding Signal Operation} \\
\hline
%\midrule
\addlinespace
$\bigl[ \begin{smallmatrix} 1/\alpha& 0  & \vline & & {0}   \\
 0 & \alpha & \vline& &{0} \end{smallmatrix} \bigr] = \pmb\Lambda_\alpha $				&  \textbf{Time Scaling} \\	[2.5pt] 

$\bigl[ \begin{smallmatrix} 1 & 0  & \vline & & {\tau}   \\
 0 & 1 & \vline& &{0} \end{smallmatrix} \bigr] = \pmb\Lambda_\tau $				&  \textbf{Time Shift} \\	[2.5pt]

$\bigl[ \begin{smallmatrix} 1 & 0  & \vline & & {0}   \\
 0 & 1 & \vline& &{\xi} \end{smallmatrix} \bigr] = \pmb\Lambda_\xi $				&  \textbf{Frequency Shift/Modulation} \\	[3pt] 

\addlinespace
\hline
\rowcolor{yellow!40} 
{ SAFT Parameters} $\left(\bsls \right) $ & { Corresponding Optical Operation} \\
\hline
%\midrule
\addlinespace
$\bigl[ \begin{smallmatrix} &\cos\theta&\sin\theta  & \vline & & {0}   \\
 -&\sin\theta&\cos\theta & \vline& &{0} \end{smallmatrix} \bigr] = \pmb\Lambda_\theta $			&  \textbf{Rotation} \\	[2.5pt] 

$\bigl[ \begin{smallmatrix} 1 & 0  & \vline & & {0}   \\
 \tau & 1 & \vline& &{0} \end{smallmatrix} \bigr] = \pmb\Lambda_\tau $				&  \textbf{Lens Transformation} \\	[2.5pt] 
  
 $\bigl[ \begin{smallmatrix} 1 & \eta  & \vline & & {0}   \\
 0 & 1 & \vline& &{0} \end{smallmatrix} \bigr] = \pmb\Lambda_\eta $				&  \textbf{Free Space Propagation} \\	[2.5pt] 
 
 $\bigl[ \begin{smallmatrix} e^{\beta} & 0  & \vline & & {0}   \\
 0 & e^{-\beta}  & \vline& &{0} \end{smallmatrix} \bigr] = \pmb\Lambda_\beta $				&  \textbf{Magnification} \\	[2.5pt] 
 
  $\bigl[ \begin{smallmatrix} \cosh\alpha & \sinh\alpha  & \vline & & {0}   \\
\sinh\alpha & \cosh\alpha  & \vline& &{0} \end{smallmatrix} \bigr] = \pmb\Lambda_\eta $				&  \textbf{Hyperbolic Transformation} \\	[2pt] 
\bottomrule
\end{tabular*}
\label{tab:1}
\end{table}

\section{The Special Affine Fourier Transform}
\label{sec:SAFT}
The Special Affine Fourier Transform or the SAFT was introduced by Abe and Sheridan \cite{Abe:1994} as a generalization of the FrFT. The SAFT can be thought of as a {versatile} transformation which parametrically generalizes a number of well known {unitary and non-unitary} transformations as well as mathematical and optical operations. In Table~\ref{tab:1} we list its parameters together with the associated mappings. 

\subsection{Forward Transform}
Mathematically, the SAFT of a signal $f(t)$ is a mapping, $\mathscr{T}_{\textrm{SAFT}} : f \to \widehat{f}_{\bsls}$
%\begin{equation}
%\label{fwd}
%
%\end{equation}
 which is defined by an integral transformation parameterized by a matrix $\bsls$
%%%%%%%%%%%%%%%%%%%%%%%%%%
%\begin{empheq}[box={\YellowboxSmall[\color{red} \scriptsize \textsf{Special Affine Fourier Transform (SAFT)}]}]
\begin{align}
\label{saft}
\mathscr{T}_{\bsls} \left[ f \right] & = {\widehat f_{\bsls} }\left( \omega  \right) =  \begin{cases}
  {\left\langle {f,{\kappa _{\bsls} }\left( { \cdot ,\omega } \right)} \right\rangle }&{b \ne 0} \\
  {\sqrt d {e^{\j\frac{{cd}}{2}{{\left( {\omega  - p} \right)}^2} + \j\omega q}}f\left( {d\left( {\omega  - p} \right)} \right)}&{b = 0}.
 \end{cases}
\end{align}
%%%%%%%%%%%%%%%%%%%%%%%%%%
{When $b\neq0$}, the matrix $\bsls^{(2\times3)}$ is the SAFT parameter matrix,
\begin{equation}
\label{abcd}
\bsls = \left[ {\begin{array}{*{20}{c}}
  a&b&\vline & p \\ 
  c&d&\vline & q 
\end{array}} \right] \equiv   \left[ {\begin{array}{*{20}{c}}  \bsll &\vline & \nvec{{\lambda}  } \end{array}} \right],
\end{equation}
which is obtained by concatenating the Linear Canonical Transform or the LCT matrix, 
\[\bsl_\textsf{L} = \left[ {\begin{array}{*{20}{c}}
  a&b \\ 
  c&d 
\end{array}} \right], \quad |\bsl_\textsf{L}| = 1 
{\mbox{ or } ad - bc = 1}
\]
(see Table~\ref{tab:1} and \cite{Moshinsky:1971}), and, an offset vector, $\nvec{\lambda}  = \left[ p \ \ q \right]^\top$ with elements $p$ and $q$ {that} represent displacement and modulation, respectively. {Let $\nvec{r} = [t \ \ \omega]^\top$ denote the time-frequency co-ordinates. The function $\kappa_{\bsls} \left(\nvec{r}\right)$ in \eqref{saft} is the parametric SAFT kernel based on a complex exponential of quadratic form,
\begin{align}
\label{saftkernel}
%& {\kappa_{\bsls} }\left( {t,\omega } \right)  \notag \\ 
%&= \Ki\exp \left( { - \frac{\j}{{2b}}\left( {a{t^2} + d{\omega ^2} + 2t\left( {p - \omega } \right) - 2\omega \left( {dp - bq} \right)} \right)} \right)
{\kappa_{\bsls} }\left( \nvec{r} \right) & = \Ki \exp \left( { - \jmath \left( {{{\mathbf{r}}^ \top }{\mathbf{Ur}} + {{\mathbf{v}}^ \top }{\mathbf{r}}} \right)} \right)
\end{align}
where, 
\[{\mathbf{U}} = \frac{1}{{2b}}\left[ {\begin{array}{*{20}{c}}
  a&{ - 1} \\ 
  { - 1}&d 
\end{array}} \right],{\mathbf{v}} = \frac{1}{b}\left[ {\begin{array}{*{20}{c}}
  p \\ 
  {bq - dp} 
\end{array}} \right]{\text{ and }}\K = \frac{1}{{\sqrt {\jmath 2\pi b} }}\exp \left( {\jmath\frac{{{dp^2}}}{{2b}}} \right).
\]
Both $\mathbf{U}$ and $\nvec{v}$ are parameterized by $\bsls$ and hence the SAFT kernel is also parameterized by $\bsls$. The exponential part of the kernel is explicitly written as, 
\begin{align*}
%& {\kappa_{\bsls} }\left( {t,\omega } \right)  \notag \\ 
 & \exp \left( { - \jmath \left( {{{\mathbf{r}}^ \top }{\mathbf{Ur}} + {{\mathbf{v}}^ \top }{\mathbf{r}}} \right)} \right) = \exp \left( { - \frac{\jmath}{{2b}}\left( {a{t^2} + d{\omega ^2} + 2t\left( {p - \omega } \right) - 2\omega \left( {dp - bq} \right)} \right)} \right).
\end{align*}
}

{Note that $\bsl_\mathsf{L}$ has $3$ free parameters $\{a,b,d\}$ and $c$ is constrained by $|\bsl_\textsf{L}| = 1$.} Due to this concatenation of the LCT matrix with a vector, the SAFT is also referred to as the \textit{Offset Linear Canonical Transform} or the OLCT \cite{Pei:2007}. {The matrix $\bsls$ arises naturally in applications involving optics and imaging. We refer the reader to the books \cite{Gerrard:1975,Healy:2016} for further details on the intuitive meaning of such a matrix representation.} 

%$$ \K = \frac{1}{{\sqrt {\jmath 2\pi b} }}\exp \left( {\jmath\frac{{{dp^2}}}{{2b}}} \right).$$

%
The SAFT of the Dirac distribution is calculated by,
\begin{align}
\label{eqn:diracsaft}
\widehat \delta_{\bsls} \left( \omega  \right)  & \EQc{saft}  \left\langle {\delta ,\kappa_{\bsls} \left( { \cdot ,\omega } \right)} \right\rangle \notag \\
& = {\left. {{\kappa_{\bsls} ^*}\left( {t,\omega } \right)} \right|_{t = 0}} \notag \\
& = {K_b}\exp \left( {\tfrac{\jmath }{{2b}}\left( {d{\omega ^2} - 2\omega \left( {dp - bq} \right)} \right)} \right),
\end{align}
and is non-bandlimited.

\subsection{Inverse Transform}
In order to define the inverse-SAFT, we first note that the SAFT satisfies the following composition property, 
\begin{equation}
\label{P1}
\l \saftC{2}\circ\saftC{1} \r \left[ f \right] = z_0 \saftL{f}{3} 
\end{equation}
where $z_0$ is {a} complex number (phase offset). The elements of the resultant SAFT parameter matrix are specified by, 
%
%%%%%%%%%%%
\begin{equation*}
%\label{P2}
\bsl_{{\sf S}_3}  = \left[ {\begin{array}{*{20}{c}}  \bsl_{{\sf L}_3}&\vline & \nvec{{\lambda}}_3 \end{array}} \right] = \left[ {\begin{array}{*{20}{c}}
  \bsl_{{\sf L}_2} \bsl_{{\sf L}_1} & \vline &  \bsl_{{\sf L}_2}\nvec{{\lambda}}_2 + \nvec{{\lambda}}_1
\end{array}} \right].
\end{equation*}
%%%%%%%%%

In {the} context of phase space, the physical significance of the SAFT parameter matrix is that it maps time-frequency co-ordinates $\nvec{r} = [t \ \ \omega]^\top$ into its affine transformed version, 
\[\left[ {\begin{array}{*{20}{c}}
  t \\ 
  \omega  
\end{array}} \right]\xrightarrow{{{\sf{SAFT}}}}\left[ {\begin{array}{*{20}{c}}
  a&b \\ 
  c&d 
\end{array}} \right]\left[ {\begin{array}{*{20}{c}}
  t \\ 
  \omega  
\end{array}} \right] + \left[ {\begin{array}{*{20}{c}}
  p \\ 
  q 
\end{array}} \right] \  \equiv \ \nvec{r} \xrightarrow{{{\sf{SAFT}}}} \bsl_{{\sf L}}\nvec{r} + \lambda.
\]
%%% 
%%%
Hence, the inverse-SAFT is defined by some affine transform that allows for the mapping, 
\[\left[ {\begin{array}{*{20}{c}}
  {at + bp + p} \\ 
  {ct + d\omega  + q} 
\end{array}} \right]\xrightarrow{{{\sf{Inverse~SAFT}}}}\left[ {\begin{array}{*{20}{c}}
  t \\ 
  \omega  
\end{array}} \right].\]
%%%%%%%%%%%%%%%%
Thanks to the composition property (\ref{P1}), setting, 
\begin{align*}
\bsl_{{\sf S}_3} & = 
\begin{bmatrix}
\bsl_{{\sf L}_3} & \nvec{\lambda}_3
\end{bmatrix}
=
\begin{bmatrix}
\mathbf{I} & \nvec{\mathbf{0}}
\end{bmatrix} \qquad \mbox{(Identity Operation)}\\
& \Rightarrow \bsl_{{\sf L}_2} = \bsl_{{\sf L}_1}^{-1} \quad \mbox{ and } \quad \nvec{{\lambda}}_2  = - \bsl_{{\sf L}_2}\nvec{{\lambda}}_1,
%\begin{cases} & \bsl_{{\sf L}_2} = \bsl_{{\sf L}_1}^{-1} \\ & \nvec{{\lambda}}_2  = - \bsl_{{\sf L}_2}\nvec{{\lambda}}_1,
%\end{cases}
\end{align*}
results in the inverse parameter matrix defining inverse-SAFT which is equivalent to an SAFT with matrix $\bslsi$,
\begin{equation}
\label{invmat}
\bslsi =  \left[ {\begin{array}{*{20}{c}}
  { + d}&{ - b}&\vline & {bq - dp} \\
  { - c}&{ + a}&\vline & {cp - aq}
\end{array}} \right] =
\begin{bmatrix}
\bsl_{{\sf L}}^{-1} &  - \bsl_{{\sf L}}^{-1}\nvec{{\lambda}}
\end{bmatrix}
\end{equation}
{where $c = \frac{{ad - 1}}{b}$}. 
Thus, the inverse transform (iSAFT) is defined as an SAFT with matrix $\bslsi$ in (\ref{invmat}), 
%\begin{empheq}[box={\YellowboxSmall[\color{red} \scriptsize \textsf{Inverse Special Affine Fourier Transform (SAFT)}]}]
\begin{align}
\label{iSAFT}
\iS [\widehat f ]  & = f\l t\r = {\CB}\left\langle {{{\widehat f}_{\bsls}},{\kappa _{\bslsi}}\left( { \cdot ,t} \right)} \right\rangle
\end{align}
where ${\kappa _{\bslsi}}\left( {\omega ,t} \right) = \kappa^*_{\bsls}\l t,\omega \r$ and, % $\CB$ is some constant.
{$$\CB = \exp \left( {\frac{\jmath }{2}\left( {cd{p^2} + ab{q^2} - 2adpq} \right)} \right).$$}

\subsection{Geometry of the Special Affine Fourier Transform}

An intriguing property of the SAFT is its geometrical interpretation in the context of time-frequency representations and the fact that the parameter matrix belongs to a class of area preserving matrices---the ones whose determinant is unity. We elaborate on these aspects starting with the cyclic property of the Fourier transform \cite{Condon:1937}.

\begin{figure}[!t]
\centering
\includegraphics[width =0.75\textwidth]{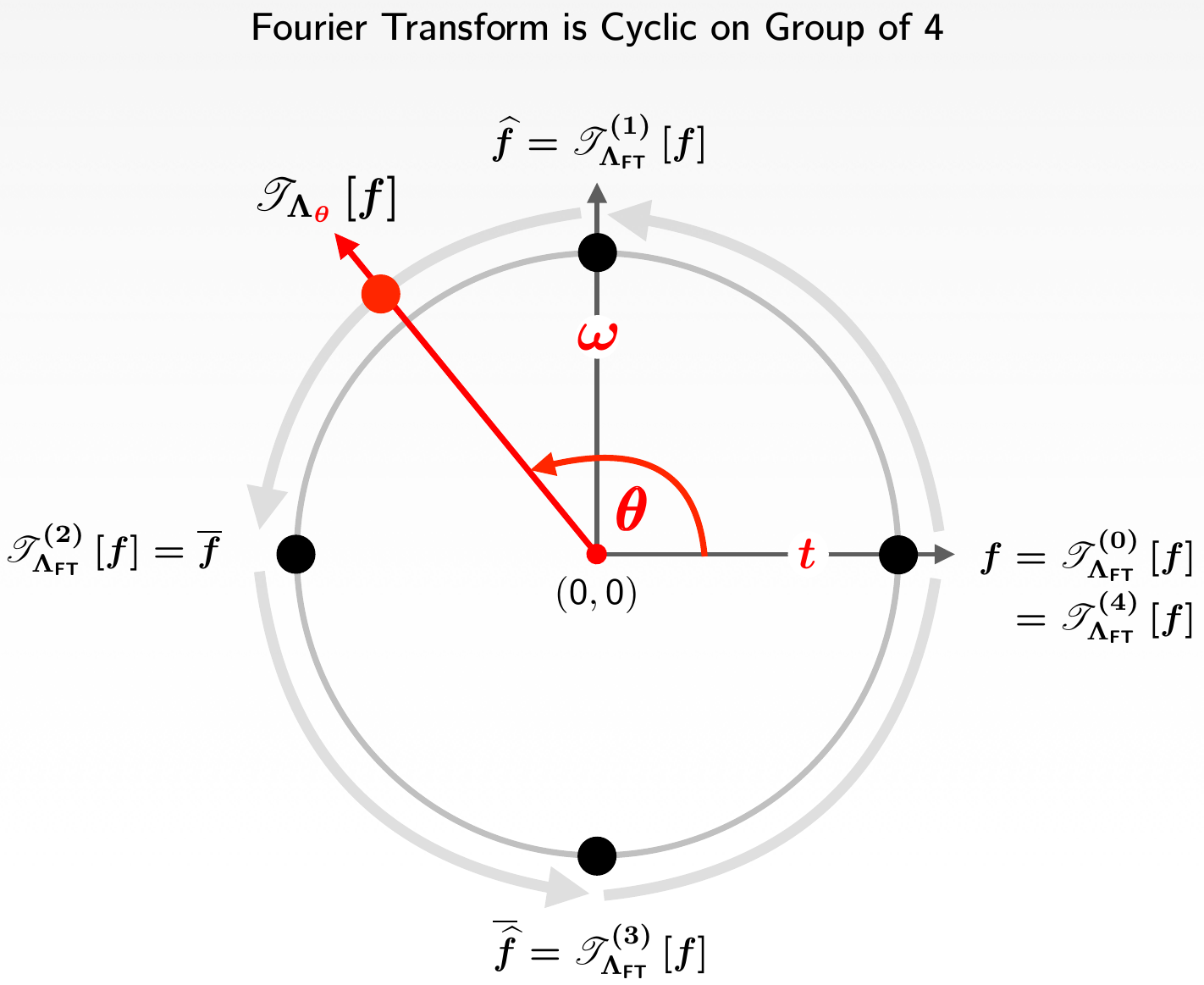}
\caption{Fourier transform is cyclic on a group of $4$, that is, $\mathscr{T}^{\l k + 4 \r}_{ \boldsymbol{\Lambda}_{{\mathsf {FT}}} } = \mathscr{T}^{\l k \r}_{ \boldsymbol{\Lambda}_{{\mathsf {FT}}} }, k\in\mathbb{Z}^+ $ as described in (\ref{FTComp}). {The} Fractional Fourier transform allow{s} for ``fractionalization'' of $k$ so that a version of the Fourier transform may be defined on an arbitrary point on the circle. We denote such as transform {by} $\mathscr{T}_{\bsl_{\theta}}\left[f \right], \theta \in \mathbb{R}$.}
\label{fig:FrFT}
 \end{figure}

Let $\mathscr{T}^{\l 0 \r}_{ \boldsymbol{\Lambda}_{{\mathsf {FT}}} } = \mathrm{I}$ be the identity operation that is, $\mathscr{T}^{\l 0 \r}_{ \boldsymbol{\Lambda}_{{\mathsf {FT}}} } [f] = f$ which we use to define the Fourier operator composition: 
\begin{equation}
\label{FTComp}
\mathscr{T}^{\l k \r}_{ \boldsymbol{\Lambda}_{{\mathsf {FT}}} } = \mathscr{T}^{\l k-1 \r}_{ \boldsymbol{\Lambda}_{{\mathsf {FT}}} } \circ \mathscr{T}_{ \boldsymbol{\Lambda}_{{\mathsf {FT}}} } = \underbrace{\mathscr{T}_{ \boldsymbol{\Lambda}_{{\mathsf {FT}}} } \circ \cdots \circ \mathscr{T}_{ \boldsymbol{\Lambda}_{{\mathsf {FT}}} }}_{k-\sf{times}}, \quad k \in \mathbb{Z}^+.
\end{equation}
Note that:
%\[
%\widehat f = \FTN{f}{1} \quad \bar{f} = \FTN{f}{2} \quad \cdots \quad f = \FTN{f}{4} \equiv %\FTN{f}{0}.
%\]
\[\arraycolsep=0.7cm
\begin{array}{*3l}
  {\widehat f = \FTN{f}{1}} &{\overline{f} = \FTN{f}{2}} \\ [4pt]
  {\overline{\widehat{f}} = \FTN{f}{3}}& \underbrace{f = \FTN{f}{4} \equiv \FTN{f}{0}}_{\sf{Identity \ Operation}}.
\end{array}\]
From the last equality, $\FTN{f}{4} = \FTN{f}{0}$, we conclude that the Fourier operator is periodic with $N=4$. Due to this periodic structure, the Fourier operator can be represented on a circle as shown in Fig.~\ref{fig:FrFT}.

 \begin{figure*}[!t]
    \centering
    \includegraphics[width =1\textwidth]{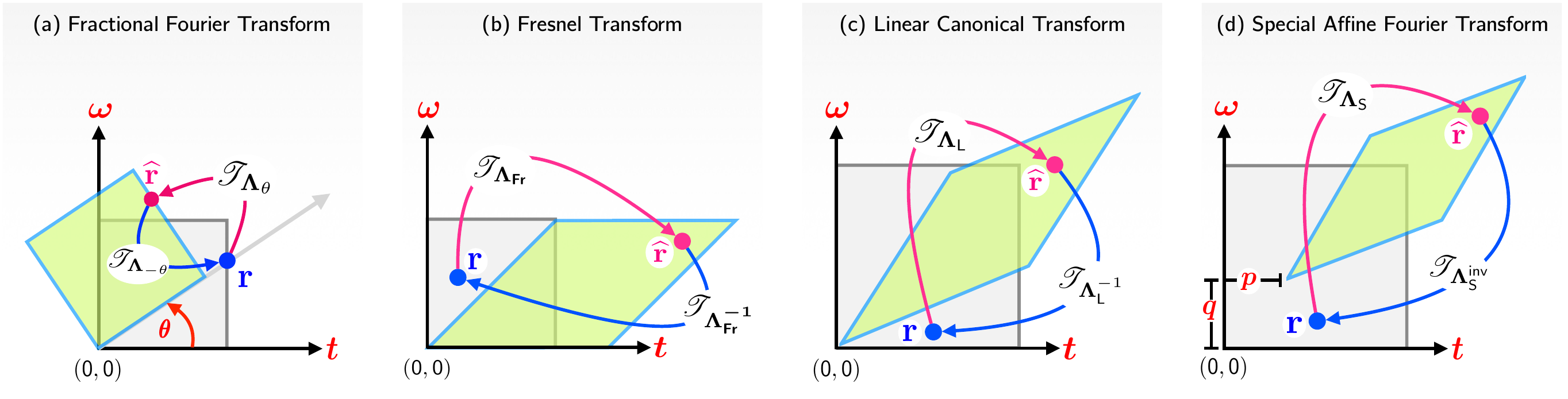}
    \caption{The SAFT maps one convex enclosure into another while preserving area since the transformation matrix $\bsll \in \SL$. The action of $\bsll$ on a time-frequency co-ordinate $\nvec{r} \in \mathbb{R}^2\l\left[ {0,1} \right]\r$ results in $\widehat{\nvec{r}} = \bsll\nvec{r}$. (a) {For the} FrFT, $\bsll = \bsl_{\sf \theta}$ implements rotation. The inverse transform {corresponds to} $\bsl_\theta^{-1}=\bsl_\theta^{\top}$. (b) {For the} Fresnel transform, $\bsll = \bsl_{\sf Fr}$ implements shear. The inverse transform {corresponds to} $\bsl_{\sf Fr}^{-1}$. (c) {For the} Linear Canonical Transform (LCT), $\bsll$ deforms a unit square into an arbitrary parallelogram. The inverse transform {corresponds to $\bsll^{-1}$.} (d) For the SAFT, presence of an offset $\nvec{\lambda} = [p \ \ q]^\top$ results in an affine transform. Consequently, the inverse transform {is} $\nvec{r} = \bsll^{-1} \widehat{\nvec{r}} -  \bsll^{-1}\nvec{\lambda}$.}
    \label{fig:SAFT}
 \end{figure*}

Unitary mappings that can be continuously defined on the circle (as opposed to $k \in \mathbb{Z}^+$) {were} first identified by Condon \cite{Condon:1937}. {This is known as the fractional Fourier transform (FrFT). Qualitatively, the FrFT ``fractionalizes'' the Fourier transform in the sense that $\FTN{f}{k}, k\in\mathbb{Z}^+$ can be defined for an arbitrary point on the circle through $\theta = k \pi/2, k\in \mathbb{Z}$ by the transformation $\mathscr{T}_{\bsl_{\theta}}\left[f \right], \theta \in \mathbb{R}$. We compare the Fourier transform with the FrFT in Fig.~\ref{fig:FrFT}.} As shown in Fig.~\ref{fig:SAFT}(a), the action of the FrFT on the time-frequency co-ordinates $\nvec{r} = [t \ \ \omega]^\top, \nvec{r}  \in \mathbb{R} ( \left[ 0 , 1 \right]^2) $ results in rotation of the time-frequency plane \cite{Bastiaans:1998} due to $\nvec{\widehat r} = \bsl_{\theta} \nvec{r}, \theta \in \mathbb{R}$---an intrinsic property of the FrFT. This is explained by the co-ordinate transformation matrix---the rotation matrix $\bsl_{\theta}$ in case of the FrFT (cf~Table~\ref{tab:1}). 

{The submatrix $\bsll$ of $\bsls$ may be decomposed in several ways. One interesting decomposition relates $\bsll$ to the Fourier transform such that $\bsll =\mathbf{M}_1 \bsl_{\sf FT } \mathbf{M}_2$ where
\[{{\mathbf{M}}_1} = \left[ {\begin{array}{*{20}{c}}
  b&0 \\ 
  d&{1/b} 
\end{array}} \right]{\text{ and }}{{\mathbf{M}}_2} = \left[ {\begin{array}{*{20}{c}}
  1&0 \\ 
  {a/b}&1 
\end{array}} \right]\]
are modulation matrices\footnote{We refer to $\mathbf{M}_1$ and $\mathbf{M}_2$ as modulation matrices because whenever $b=0$ in \eqref{abcd}, the SAFT in \eqref{saft} amounts to modulation of the function $f$.}. This decomposition implies that the SAFT can be implemented as a Fourier transform using the following sequence of steps, 
\begin{align*}
f_1\l x \r  & \DE  \mathscr{T}_{\mathbf{M_2}|\nvec{0}}
\left[ f \right]\l x \r
\EQc{saft} 
{e^{\jmath \frac{a}{{2b}}{x^2}}}f\left( x \right),   \\
\widehat{f_1}\l \xi \r & = \mathscr{T}_{\bsl_{\mathsf{FT}}}  \left[ f_1 \right]\l \xi \r  = \int {{f_1}\left( x \right){e^{ - \jmath \xi x}}dx},\\
\widehat{f}_1\l \omega \r  & \DE  \mathscr{T}_{\mathbf{M_1}|\boldsymbol{\lambda}} 
[ \widehat{f}_1]\l \omega \r 
\EQc{saft} 
\frac{1}{{\sqrt b }}{e^{\jmath \omega q + \jmath \frac{d}{2b}{{\left( {\omega  - p} \right)}^2}}}{\widehat f_1}\left( {\frac{{\omega  - p}}{b}} \right).
\end{align*}
By simplifying $\widehat{f}_1\l \omega\r$, we observe that it is indeed the SAFT of $f\l t \r$. In this way, we generalize the previously known result of Zayed \cite{Zayed:1996} that links the FrFT to the Fourier Transform.}

{An alternative} decomposition relates the SAFT with the FrFT and the Fresnel transform via the elegant Iwasawa Decomposition,
\[ \bsll = \bsl_{\theta}\left[ {\begin{array}{*{20}{c}}
  \Gamma &0 \\ 
  0&{{\Gamma ^{ - 1}}} 
\end{array}} \right] 
\underbrace{\left[ {\begin{array}{*{20}{c}}
  1&u \\ 
  0&1 
\end{array}} \right]}_{\sf{Fresnel~Transform}}, %
\begin{array}{*{20}{l}}
  {\Gamma  = \sqrt {{a^2} + {c^2}} } \\ 
  {u = \left( {ab + cd} \right)/{\Gamma ^2}} 
\end{array}.
\] 
%Geometrical interpretation of all the matrices listed in Table~\ref{tab:1}.
% is left for the readers to explore. These properties are closely tied with the problems in phase space optics \cite{Ozaktas:1993}. For example self-imaging phenomenon \cite{Pei:2002}.
%
%
{In fact, rotation is a special operation of a class of matrices that belong to the special linear group $\SL$ where, 
\[\SL = \left\{ {{\mathbf{A}} = \left[ {\begin{array}{*{20}{c}}
  {{a_1}}&{{a_2}} \\ 
  {{a_3}}&{{a_4}} 
\end{array}} \right]
\left\{ {{a_k}} \right\}_{k = 1}^4 \in \mathbb{R}{\text{ and }}\left| {\mathbf{A}} \right| = 1} \right\}.\]
With the exception of the Laplace, Gauss and Bargmann transforms in Table~\ref{tab:1}, all other operations can be explained by $\bsll \in \SL$ which entails that $ad-bc=1$. Since the basis vectors of $\bsll$ form a parallelogram in $\mathbb{R}^2$, its enclosed area must always be unity or the area must be preserved under application of $\bsll$.} This aspect has important consequences in ray optics where $\bsl_{\sf{L}}$ models paraxial optics \cite{Abe:1994,Gerrard:1975}. In Figs.~\ref{fig:SAFT}(b), \ref{fig:SAFT}(c) and \ref{fig:SAFT}(d) we describe the deformation on $\nvec{r}$ due to $\bsls$ for the Fresnel transform $\l\bsls = \bsl_{\sf{Fr}}\r$, the LCT $\l\bsls = \bsl_{\sf{L}}\r$ and the SAFT. 

Geometrically, the inverse transform relies on specification of $\bslsi$ which undoes the effect of $\bsls$. For the FrFT, the Fresnel transform and the LCT, the operation is simply the inverse of the matrix, that is $\bslsi = \bsls^{-1}, \bsls = \{\bsl_\theta,\bsl_{\sf Fr}, \bsll\}$ (cf. Table~\ref{tab:1}). The case of the SAFT is unique because it implements an affine transform as opposed to the usual case of a linear transform (cf. compare Fig.~\ref{fig:SAFT}(b,c) and Fig.~\ref{fig:SAFT}(d)). The presence of an offset $\nvec{\lambda} = [p \ \ q]^\top$ in (\ref{abcd}) warrants an adjustment by  $- \bsl_{{\sf L}}^{-1}\nvec{{\lambda}}$ (\ref{invmat}) {for} the SAFT.

\subsection{Convolution Structures in the SAFT Domain}
\label{subsec:SAFTConv}
A useful property of the Fourier transform is that the convolution of two functions is equal to the pointwise multiplication of their spectrums. More precisely, $\FT{f*g} = \FT{f}\FT{g}$. However, this property does not extend to the SAFT domain in that $\SAFT{f*g} \neq \SAFT{f}\SAFT{g}$ (cf. \cite{Xiang:2012}). Since convolutions are pivotal to the topic of sampling theory, we will work with a generalized version of the convolution operator, denoted by $\SC$, which allows for a representation of the form $\SAFT{f\SC g} \propto \SAFT{f}\SAFT{g}$.

\begin{definition}[Chirp Modulation] Let $\mathbf{A}=[a_{j,k}]$ be a $2\times 2$ matrix. We define the chirp modulation function as,
\begin{equation}
\label{CP}
m_{\mathbf{A}} \l t \r \DE \exp\l \j \frac{ a_{1 1}}{ 2 a_{12}}  t^2 \r.
\end{equation}
We also define {the} $\mathbf{A}$--parametrized unitary up and down chirp modulation operation, 
\begin{equation}
\up{f}{t}  \DE   m_{\mathbf{A}} \l t \r f\l t\r \ \ \mbox{ and } \ \
\dn{f}{t}  \DE   m^*_{\mathbf{A}} \l t \r f\l t\r, 
\label{cmod}
\end{equation}
respectively. Note that ${\| \up{f}{t}\|^2_{{L_2}}} = {\left\| f\l t\r \right\|^2_{{L_2}}}$.
\label{def:mod}
\end{definition}

Based on the definition of chirp modulated functions, we now define the SAFT convolution operator. 

\begin{definition}[SAFT Convolution] Let $*$ denote the usual convolution operator. Given functions $f$ and $g$, the SAFT convolution operator denoted by $\SC$, is defined as 
\begin{equation}
\label{sconv}
\RB{h}{t} = \RB{\l f \SC g\r}{t} \DE K_b \RB{m^*_{\bsls}}{t} \RB{\l{\upo{f} * \upo{g}}\r}{t},
\end{equation}
where $\up{f}{t} = \RB{m_{\bsls}}{t}\RB{f}{t} \DEq{cmod} e^{\j \tfrac{at^2}{2b}}\RB{f}{t}$; the same applies to the function $g$.
\label{def:saftconv}
\end{definition}
\begin{figure}[!t]
    \centering
    \includegraphics[width =0.7\textwidth]{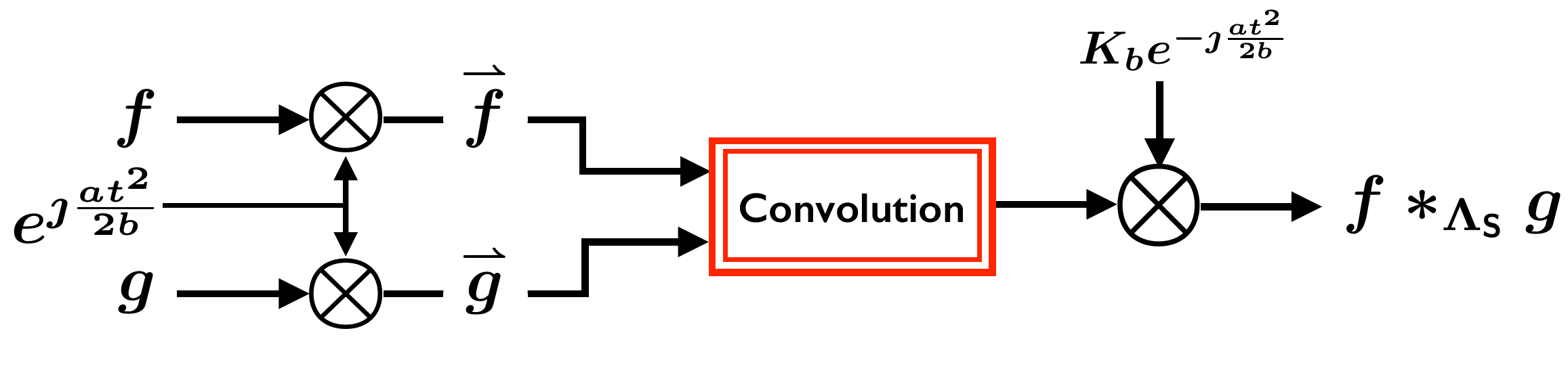}
    \caption{Block diagram for SAFT convolution. We use the usual definition of the convolution operation in conjunction with chirp modulation defined in (\ref{cmod}) to define the SAFT convolution operation.}
    \label{fig:SAFTConv}
 \end{figure}
In Fig.~\ref{fig:SAFTConv}, we explain the SAFT convolution operation defined in (\ref{sconv}). Note that the SAFT convolution operation is based on the usual convolution of pre-modulated functions $\up{f}{t}$ and $\up{g}{t}$. {This operation, also known as \emph{chirping}, is a standard procedure in optical information processing \cite{Ozaktas:1994} and analog processing where it is implemented via mixing circuits (cf.~Fig.~6(a) in \cite{Harms:2015a}). Similarly, in the field of holography, such operations are used for defining Fresnel transforms (cf.~(10) in \cite{Chacko:2013}).} Pre- and post-modulations are critical in our context and enforce the convolution-multiplication property. A formal statement of this result is as follows:
\begin{theorem}[Convolution and Product Theorem \cite{Bhandari:2017,Bhandari:2018}]
\label{thm:saftconv}
Let $f$ and $g$ be two given functions and let $ h\l t \r = \l f\SC g\r \l t \r$ be defined in \eqref{sconv}. Then, 
\[
h\l t \r \xrightarrow{{{\mathsf{SAFT}}}}{\saft{h} =\PH\l \omega \r \saft{f}\saft{g}},
\]
where $\saft{f}, \saft{g}$ and $\saft{h}$ denote the SAFT of $f,g$ and $h$, respectively and $\PH \left( \omega  \right) = {e^{\jmath \frac{\omega }{b}\left( {dp - bq} \right)}}{e^{-\jmath \frac{{d{\omega ^2}}}{{2b}}}}$.
\end{theorem}
{The proof of this theorem is presented in \cite{Bhandari:2017}. For further results, we the reader to \cite{Bhandari:2018}.} The duality principle also holds for the SAFT. {Namely}, multiplication of functions in the time domain results in convolution in the SAFT domains,
\begin{equation*}
{\PHI}\l t \r f\l t\r g \l t \r \xrightarrow{{{\mathsf{SAFT}}}} \CB \l \widehat{f} \SCI \widehat{g} \r \l \omega \r
\end{equation*}
where $\CB$ is defined in (\ref{iSAFT}). This result is based on the proof of Theorem~\ref{thm:saftconv}. For {further details}, we refer the reader to \cite{Bhandari:2018}.

\section{SAFT Domain and Bandlimited Subspaces}

\label{sec:ShannonSAFT}

In order to set the ground for sampling of sparse signals, {we begin by recalling the sampling theorem for SAFT bandlimited signals \cite{Bhandari:2016,Bhandari:2017}.} The notion of bandlimitedness has a de facto association with the Fourier domain. {Below}, we consider a more general definition. 
\begin{definition}[Bandlimited Functions] Let $f$ be a square-integrable function. We say {that} $f$ is $\Omega_m$--bandlimited {and write}, 
\[f \in \mathcal{B}_{\bsls}^{\Omega_m}  \Leftrightarrow f\left( t \right) = \CB\int\nolimits_{ - \Omega_m}^{ + \Omega_m } {\saft{f}\kappa^*_{\bslsi} \left( \omega  \right)d\omega }.\]
With $\bsls = \bsl_{\sf FT }$, we obtain the standard case when $f$ is $\Omega_m$--bandlimited in the Fourier domain. 
\end{definition}
Shannon's sampling theorem is restricted to Fourier transforms. In that case, $\bsls = \bsl_{\sf FT}$ and any $f \in  \mathcal{B}_{\bsl_{\sf FT}}^{\Omega_m}$ can be uniquely recovered from samples $f\l k\Delta \r , k \in\mathbb{Z}$ provided that $\Delta \leq \pi/\Omega_m$. For bandlimited signals in the SAFT sense, the statement of Shannon's sampling theorem is follows. 

\begin{theorem}[Shannon's Sampling Theorem for the SAFT Domain \cite{Bhandari:2017}] Let $f$ be an $\Omega_m$--bandlimited function in the SAFT domain, that is, $f \in  \mathcal{B}_{\bsls}^{\Omega_m} $. Then, we have, 
\[
f\l t \r = {e^{ - \jmath \frac{{a{t^2}}}{{2b}}}}\sum\limits_{n \in \mathbb{Z}} {\up{f}{n\Delta}{e^{ - \jmath p\frac{{t - n\Delta}}{b}}}\sincD{t-n\Delta}}.
 \]
\end{theorem}
{A detailed proof of this theorem that is based on reproducing kernel Hilbert spaces is given in \cite{Bhandari:2017}. Here, we will briefly revisit the key steps. The associated computations will be useful in the context of sampling sparse signals.}

{It is well known that} the sampling theorem for the Fourier domain can be interpreted as an orthogonal projection of $f$ onto the subspace of bandlimited functions \cite{Eldar:2015},
\begin{equation}
\label{VBL}
\vbl = {\text{span}}{\left\{ \frac{1}{\sqrt{\Delta}} \sincD{t-n\Delta} \right\}_{n \in \mathbb{Z}}}.
\end{equation}
Thanks to the projection theorem,
\begin{equation}
\label{proj}
{\mathscr{P}_{\vbl}}f = \arg \mathop {\min }\limits_{g \in {\sf V_{{\sf{BL}}}}} \left\| {f - g} \right\|_{{L_2}}^2 \ , \ f \in  \mathcal{B}_{\bsl_{\sf FT}}^{\Omega_m}  \Leftrightarrow  f = {\mathscr{P}_{\vbl}}f.
\end{equation}
In the spirit of the Fourier domain result, in \cite{Bhandari:2017}, we derived the subspace of bandlimited functions linked with the SAFT domain which take the form of, 
\begin{equation}
\label{gsinc}
{\phi _n}\left( t \right) =
\frac{{{e^{ - \jmath \frac{p}{b}\left( {t - n\Delta } \right)}}}}{{\sqrt \Delta  }}
m_{\bsls}^*\left( t \right){m_{\bsls}}\left( {n\Delta } \right)\sincD{t-n\Delta}.
\end{equation}
The family ${\left\{ {{\phi _n}\left( t \right)} \right\}_{n \in \mathbb{Z}}}$ is an orthonormal basis for the subspace of bandlimited functions in the SAFT domain. Indeed, $\phi_0 \in \mathcal{B}_{\bsls}^{\Omega_m}$ with $\Omega_m = \pi b / \Delta$ since, 
\begin{equation*}
\label{ }
\saft{\phi_0} = {\sqrt\Delta} {\K} \PH^*\l \omega \r \underbrace{\ind_{[-\Omega_m , \Omega_m]}\l \omega \r}_{\sf{Bandlimited}} \ , \ \Omega_m = \frac{\pi b}{\Delta}.
\end{equation*}
Thanks to the orthonormality and the bandlimitedness properties, the implication of the projection theorem (cf.~\eqref{proj}) is that $f \in  \mathcal{B}_{\bsls}^{\Omega_m}  \Leftrightarrow  f = {\mathscr{P}_{\vbl}}f$ and {by developing this further, we obtain,} 
\begin{align}
\label{shannonSAFT}
 {\mathscr{P}_{\vbl}}f &= \sum\limits_{n \in \mathbb{Z}} {\left\langle {f,{\phi _n}} \right\rangle {\phi _n}\left( t \right)}\\
& = {e^{ - \jmath \frac{{a{t^2}}}{{2b}}}}\sum\limits_{n \in \mathbb{Z}} {\up{f}{n\Delta}{e^{ - \jmath p\frac{{t - n\Delta}}{b}}}\sincD{t-n\Delta}}. \notag
\end{align}
As in the classical case, {the} coefficients $\{ {\left\langle {f,{\phi _n}} \right\rangle }\}_n$ are equivalent to low-pass filtering in {the} SAFT domain followed by uniform sampling. To make this link clear, consider the kernel,
\begin{equation}
\label{BLK}
{\varphi _{{\sf{BL}}}}\left( t \right) = \frac{1}{{\sqrt \Delta  {K_b}}}{m^*_{\bsls}}\left( t \right){e^{ - \jmath \frac{p}{b}t}}\sincD{-t},
\end{equation}
which is the amplitude scaled version of $\phi_0$. Using the convolution-product duality in Theorem~\ref{thm:saftconv}, {it is easy to verify \cite{Bhandari:2017}} that $\forall f \in  \mathcal{B}_{\bsls}^{\Omega_m}$, 
\begin{equation}
\label{LPP2}
\left\langle {f,{\phi _n}} \right\rangle  = \sqrt{\Delta} {\left. {\left( {f \SC {\varphi _{{\sf{BL}}}}} \right)\left( t \right)} \right|_{t = n\Delta }}, \quad \varphi _{\sf{BL}} = \K \phi_0
\end{equation}
and the expansion coefficients in \eqref{shannonSAFT} are indeed the samples.

\section{Sparse Sampling and Super-resolution}
\label{sec:main}

{When $\varphi$ is bandlimited in the Fourier domain, the super-resolution problem boils down to estimating $\{c_k\}$'s and $\{t_k\}$'s from measurements $y\l n \Delta \r, \Delta>0$ of,}
\begin{equation}
\label{sparsemodel}
y\left( t \right) = \sum\limits_{k = 0}^{K - 1} {{c_k}\varphi \left( {t - {t_k}} \right)}.
\end{equation}
This problem can be restated as {that} of sampling spikes or sparse functions given a bandlimited sampling kernel, $\varphi$. This is because of the equivalence, 
\begin{equation}
\label{LPP1}
y\left( {n\Delta } \right) = \underbrace {\int {s\left( t \right)\overline \varphi  \left( {t - n\Delta } \right)dt} }_{{\text{Projection}}} \equiv \underbrace {{{\left. {\left( {s*\varphi } \right)\left( t \right)} \right|}_{t = n\Delta }}}_{{\text{Convolution}}},
\end{equation}
where $\overline \varphi\left( t \right) = \varphi \left( -t \right)$, $\Delta>0$ is the sampling rate and $s$ is the sparse signal, 
\begin{equation}
\label{FRI}
s\left( t \right) = \sum\limits_{k = 0}^{K - 1} {{c_k}\delta \left( {t - {t_k}} \right)}.
\end{equation} 

{In the previous section, we discussed sampling theory of bandlimited signals in the SAFT domain. By considering sparse signals instead of bandlimited functions, the measurements in context of the SAFT domain amount to,}
\begin{equation}
\label{sampling}
s\xrightarrow{}\boxed\varphi \xrightarrow{} \underbrace{ s\SC\varphi  = y }_{\sf{Pre\mbox{-}Filtering}}\xrightarrow{{}} 
\underbrace{{ \otimes _{{\delta _{n\Delta }}}}\xrightarrow{}y\left( {n\Delta } \right)}_{\sf{Sampling}},
\end{equation}
where ${\otimes _{{\delta _{n\Delta }}}}$ denotes the sampling operation or modulation with a $\Delta$-periodic impulse train. {As shown in (\ref{shannonSAFT}), whenever $s\in \mathcal{B}_{\bsls}^{\Omega_m}$, samples $y\l n\Delta \r$ uniquely characterize $s$ provided that $\Delta\leq\pi b /\Omega_m$.}

Next, we turn our attention to the problem of {recovering} a sparse {signal} from low-pass projections (\ref{LPP2}). In particular, we will discus three variations on this theme where low-pass projections are attributed to:
\begin{enumerate}[leftmargin=20pt,label={{$\arabic*)$}},itemsep = 2pt]
  \item Arbitrary, bandlimited sampling kernels.
  \item Smooth, time-limited sampling kernels.
  \item Gabor functions associated with the SAFT domain. 
\end{enumerate}
{The} first two results rely on the architecture of (\ref{sampling}). The last result generalizes the recent work of Aubel \etal~\cite{Aubel:2017} and can be extended to the case of phase-retrieval \cite{Eldar:2015a,Zalevsky:1996,Dong:1997} and wavelets.

Since sparse signals are time-limited, their periodic extension allows for a Fourier series representation. That said, the basis functions of the SAFT kernel are aperiodic. As a result, before discussing the recovery of sparse signals, we introduce mathematical tools that allow for Fourier series-like representation of time-limited signals. 
\subsection{Special Affine Fourier Series (SAFS)}
\label{subsec:SAFS}
It is well known that the family of functions ${\left\{ {{e^{\jmath k{\omega _0}t}}} \right\}_{k \in \mathbb{Z}}}$, with fundamental harmonic $\omega_0 = 2\pi/T$, constitutes an orthonormal basis of $L_2\l [-\pi/T,\pi/T] \r$. These basis functions are used for representing $T$-periodic functions. Let $\vfs = \text{span}{\left\{ {{e^{\jmath k{\omega _0}t}}} \right\}_{k \in \mathbb{Z}}}$. Due to the orthonormality and completeness properties, it follows that, {for every} $f\l t \r = f\l t+T\r$,
\begin{equation}
\label{FS}
{\mathscr{P}_{\vfs}}f = \sum\limits_{n \in \mathbb{Z}} c_n e^{\jmath k \omega_0 t}, \quad c_n = \left\langle {f,e^{\jmath k \omega_0 t}} \right\rangle.
\end{equation}
Inspired by the Fourier series representation, here, we develop a parallel for the SAFT domain which is useful for the task of representing time-limited signals including sparse signals. 

In order to determine the basis functions associated with the Special Affine Fourier Series or the SAFS, we first identify the candidate functions and then enforce the orthonormality property. Note that a spike at frequency $\omega = n \omega_0$ in the SAFT domain results in the time domain function
\begin{equation}
\label{SAFTf}
\left\langle {\delta \left( {\omega  - n{\omega _0}} \right),\kappa_{\bslsi} \left( {\omega ,t} \right)} \right\rangle \DEq{iSAFT} \kappa_{\bslsi}^* \left( {n{\omega _0},t} \right).
\end{equation}
With the above as our prototype basis function, we would like to represent a time-limited signal $s\left( t \right), \ t\in [0,T )$ as
\begin{equation}
\label{safs}
s\left( t \right) = \sum\limits_{n\in\mathbb{Z}} {\widehat s_{\bsls}\left[ n \right]\kappa_{\bslsi}^* \left( {n{\omega _0},t} \right)},
\end{equation}
which mimics (\ref{FS}), and where the SAFS coefficients are
\begin{equation}
\label{safscoeff}
\widehat s_{\bsls}\left[ n \right] = {\left\langle {s,\kappa_{\bsls} \left( {\cdot,n{\omega _0}} \right)} \right\rangle _{\left[ {0,T} \right]}}.
\end{equation}

To enable a representation in the form (\ref{safscoeff}), we enforce orthogonality on the candidate basis functions,
\[{\left\langle {\kappa_{\bsls} \left( {t,n{\omega _0}} \right),\kappa_{\bslsi} \left( {k{\omega _0},t} \right)} \right\rangle _{\left[ {0,T} \right]}} \equiv \underbrace{I_{\omega_0} = w_0 \delta \left[ {n - k} \right]}_{\sf{Orthogonality}},\]
{for an appropriate $w_0$}. {Computing the inner-product explicitly yields,
\begin{align*}
\label{}
{I_{{\omega _0}}} & = {\left| \K \right|^2}\mu \int\nolimits_0^T {{e^{\jmath \frac{{{\omega _0}\left( {k - n} \right)}}{b}t}}dt} \\ 
& = {\left| \K \right|^2} \begin{cases}
   T   & k=n, \\
  - \mu  \tfrac{{\jmath b}}{{\left( {k - n} \right){\omega _0}}}\left( {1 - {e^{ - \jmath \frac{{{\omega _0}T}}{b}\left( {k - n} \right)}}} \right)   & k\neq n, 
\end{cases}
\end{align*}
{where for brevity we denoted}, 
$$\mu = \l{e^{ - \jmath \frac{{{\omega _0}\left( {k - n} \right)}}{{2b}}{ \l2 b q+d \omega_0 \l k+n\r-2 d p\r   }}}\r.$$}
Therefore, the orthogonality property will hold, if, {$${1 - {e^{ - \jmath \frac{{{\omega _0}T}}{b}\left( {k - n} \right)}}}=0 \Leftrightarrow {\omega _0} = {2\pi b}/{T}.$$}
We consolidate our result in the following definition. 
\begin{definition}[Special Affine Fourier Series] Let $s$ be a time-limited signal supported on the interval $\left[0, T\right)$. The Special Affine Fourier Series representation of $s$ is defined in (\ref{safs}) and (\ref{safscoeff}) where $\omega_0 = 2\pi b/T$. 
\end{definition}
In case of $\bsls = \bsl_{\sf{FT}}$, the SAFS {reduces to} the Fourier series. {To see this, let us substitute the parameters of $\bsl_{\sf{FT}}$, from Table~\ref{tab:1}, in $\bsls$. Then, we have, $\kappa_{\bslsi}^* \left( {n{\omega _0},t} \right)  = e^{\jmath n \omega_0 t}$ which are indeed the basis functions for the Fourier Series (upto a constant, $\K$).} Similarly, with $\bsls = \bsl_{\theta}$, the SAFS reduces to Fractional Fourier series (cf.~(10) in \cite{Pei:1999}).

Based on the definition of the SAFS for time-limited functions, we
now develop an alternative representation of sparse signals defined in (\ref{FRI}) that are supported on the interval 
$$T = \left| {\max \{t_k\}_{k=0}^{K-1} - \min \{t_k\}_{k=0}^{K-1}} \right|.$$
{With $T$ above and $\epsilon>0$, we compute the SAFS coefficients, 
\begin{align}
  \widehat s_{\bsls}\left[ n \right] & \EQc{safscoeff} {\left\langle {s,\kappa_{\bsls} \left( {t,n{\omega _0}} \right)} \right\rangle} \notag \\
 & = \int\limits_\epsilon^{\epsilon+T} s\l t \r {\kappa_{\bsls} ^*}\left( {t,n{\omega _0}} \right) dt  \notag \hfill \\
   &= \sum\limits_{k = 0}^{K - 1} {{c_k}{\kappa_{\bsls} ^*}\left( {{t_k},n{\omega _0}} \right)}, \quad \omega_0 = 2\pi b/T.  \hfill 
\label{SAFTS}
\end{align}}
{From} (\ref{safs}) we obtain the SAFT series,
\begin{equation}
\label{safsdirac}
s\left( t \right) \EQc{safs} \sum\limits_{n\in \mathbb{Z}}\sum\limits_{k=0}^{K-1} { {{c _k}{\kappa ^*_{\bsls}}\left( {{t_k},n{\omega _0}} \right)} {\kappa ^*_{\bslsi}}\left( {n{\omega _0},t} \right)}.
\end{equation}
{Using} (\ref{saftkernel}) we {have}, 
\begin{equation}
\label{aux1}
{\kappa ^*_{\bsls}}\left( {{t_k},n{\omega _0}} \right){\kappa ^*_{\bslsi}}\left( {n{\omega _0},t} \right) \EQc{saftkernel} {e^{ - \jmath \l \mathsf{Q}\left( t \right) - \mathsf{Q} \left( {{t_k}} \right)\r}}{e^{\jmath \frac{{{\omega _0}n}}{b}\left( {t - {t_k}} \right)}}
\end{equation}
where $\mathsf{Q}(t)$ is a quadratic polynomial,
\begin{equation}
\label{Qt}
\mathsf{Q}\left( t \right) \DE \left( {a{t^2} + 2pt} \right)/2b.
\end{equation}
{Substituting into (\ref{safsdirac}) leads to,}
\begin{align}
  s\left( t \right) & = {e^{ - \jmath \mathsf{Q}\left( t \right)}}\sum\limits_{n \in \mathbb{Z}} {\sum\limits_{k = 0}^{K - 1} {\underbrace {{c _k}{e^{\jmath \mathsf{Q}\left( {{t_k}} \right)}}}_{{c' _k}}\underbrace {{e^{ - \jmath \frac{{{\omega _0}n{t_k}}}{b}}}}_{u_k^n}} {e^{\jmath \frac{{{\omega _0}nt}}{b}}}} \notag \hfill \notag \\
  & = {e^{ - \jmath \mathsf{Q}\left( t \right)}}\sum\limits_{n \in \mathbb{Z}} {\widehat h \left[ n \right]{e^{\jmath \frac{{{\omega _0}nt}}{b}}}} \notag \\
  &= {e^{ - \jmath \mathsf{Q}\left( t \right)}} h\l t \r, 
\label{hm}
\end{align}
where $\widehat{h}[n]$ is a sum of complex exponentials: 
\begin{equation}
\label{SOCE}
\widehat h\left[ n \right] = \sum\limits_{k = 0}^{K - 1} {{c _k}{e^{\jmath \mathsf{Q}\left( {{t_k}} \right)}}{e^{-\jmath \frac{{{\omega _0}n}}{b}{t_k}}}}  = \sum\limits_{k = 0}^{K - 1} {{c' _k}u_k^n}. 
\end{equation}

By arranging (\ref{hm}), we may rewrite,
\begin{equation}
\label{HT}
\underbrace {s\left( t \right){e^{\jmath \mathsf{Q}\left( t \right)}}}_{h\left( t \right)} = \underbrace {\sum\limits_{n \in \mathbb{Z}} {\widehat h\left[ n \right]{e^{\jmath \frac{{2\pi }}{T}nt}}}  \equiv h\left( t \right)}_{{\sf{Fourier\ Series}}},
\end{equation}
{with $\widehat{h}\left[ n \right]$ given in (\ref{SOCE}). We have thus shown that a }modulated version of the sparse signal $s\l t \r$ is equivalent to another sparse signal,
\begin{equation}
\label{ht}
h\l t \r = \sum\limits_{n\in\mathbb{Z}}\sum\limits_{k = 0}^{K - 1} {{c'_k}\delta \left( {t - {t'_k} -nT} \right)} 
\end{equation}%
where unknowns $\{ c'_k,t'_k \}_{k=0}^{K-1}$ of $h\l t \r$ are related to the unknowns of the sparse signal we seek to recover
\[
c'_k = {{c _k}{e^{\jmath \mathsf{Q}\left( {{t_k}} \right)}}}
\quad \mbox{and} \quad t'_k = t_k.
\]
{Due {to} this link between $s\l t \r$ in \eqref{FRI} and $h\l t \r$ in \eqref{ht}, their Fourier series coefficients are also related to one another as in \eqref{HT}. We formally state this result as a theorem. 
\begin{theorem}[SAFS of Sparse Signals] 
\label{thm:safs}
Let $s\l t \r$ be the sparse signal defined in (\ref{FRI}). Then, the Special Affine Fourier Series representation of $s\l t \r$ is given by,
\[
s\left( t \right) ={e^{ - \jmath \mathsf{Q}\left( t \right)}}\sum\limits_{n \in \mathbb{Z}} {\widehat h \left[ n \right]{e^{\jmath \frac{{{\omega _0}nt}}{b}}}},
\]
where $\widehat h\left[ n \right] = \sum\limits_{k = 0}^{K - 1} {{c _k}{e^{\jmath \mathsf{Q}\left( {{t_k}} \right)}}{e^{-\jmath \frac{{{\omega _0}n}}{b}{t_k}}}} $ and $\mathsf{Q}\left( t \right)$ is defined in \eqref{Qt}.
\end{theorem}
This re-parameterization of the sparse signal $s\l t\r$ in form of the Fourier series coefficients of $h\l t \r$ is key to studying sparse sampling theorems in the SAFT domain.}

\begin{figure*}[!t]
{\small
\[\left[ {\begin{array}{*{20}{c}}
  {{g_0}} \\ 
   \vdots  \\ 
  {{g_{N - 1}}} 
\end{array}} \right] = \left[ {\begin{array}{*{20}{c}}
  1& \cdots &1& \cdots &1 \\ 
  {{e^{ - \jmath \frac{{2\pi }}{T}{f_c}}}}& \ddots &1& \ddots &{{e^{ + \jmath \frac{{2\pi }}{T}{f_c}}}} \\ 
   \vdots & \ddots &1& \ddots & \vdots  \\ 
  {{e^{ - \jmath \frac{{2\pi }}{T}{f_c}\left( {N - 1} \right)}}}& \cdots &1& \cdots &{{e^{ + \jmath \frac{{2\pi }}{T}{f_c}\left( {N - 1} \right)}}} 
\end{array}} \right]\left[ {\begin{array}{*{20}{c}}
  {\widehat h\left[ { - {f_c}} \right]} \\ 
  {} \\ 
   \vdots  \\ 
  {} \\ 
  {\widehat h\left[ {{f_c}} \right]} 
\end{array}} \right]
 \quad  \equiv \quad \nvec{g} = \bf{V} \widehat{\nvec{h}}
\tag{41}
\label{top1}
\]}
\noindent\rule{\textwidth}{0.25mm}
\end{figure*}

\subsection{Sparse Signals and Arbitrary Bandlimited Kernels}
{Consider the setting in which} the sampling kernel $\varphi =  {\varphi _{{\sf{BL}}}}\left( t \right)$ (\ref{BLK}). In this case, the low-pass filtered measurements are given by $y\l t \r = \l s \SC {\varphi _{{\sf{BL}}}}\r \left( t \right)$. {The measurements can be expressed in terms of low-pass, orthogonal projections (\ref{LPP2}) as given in the following proposition.}
\begin{proposition}[Bandlimited Case] Let $s$ be the sparse signal defined in (\ref{FRI}) and the bandlimited sampling kernel ${\varphi _{{\sf{BL}}}}\left( t \right)$ be defined in (\ref{BLK}). {Then}, the measurements simplify to
\begin{equation}
\label{BLS}
y\l t \r = \sqrt\Delta{e^{ - \jmath \mathsf{Q}\left( t \right)}}
\sum\limits_{\left| m \right| \leqslant M} {\widehat h\left[ m \right]{e^{\jmath \frac{{{\omega _0}mt}}{b}}}} ,\;\;\;{M} = \left\lfloor \tfrac{T}{ 2\Delta} \right\rfloor{,}
\end{equation}
where $\sf{Q}$ is defined in (\ref{Qt}) and $\widehat{h}[n]$ is defined in (\ref{SOCE}).
\label{prp:BLS}
\end{proposition}
We prove a more general version of this proposition in  Section~\ref{prf:prp:BLS}. Next, we state the main result linked with sampling of sparse signals in the SAFT domain. 

\begin{theorem}[Sparse Sampling with Bandlimited Kernel] Let $s \l t \r$ be a continuous--time, sparse signal (\ref{FRI}) and let $\varphi_{\sf{LP}}$ be the low--pass filter defined in (\ref{BLK}) with $\Delta = \pi b/\Omega_m$. Suppose that we observe low--pass filtered samples $y\left( n\Delta \right) = \l s\SC\varphi_{\sf{LP}}\r \l  n\Delta \r, n = 0, \ldots ,N - 1$. Provided that $K$ and $\bsls$ and known and $N\geq T/\Delta + 1$, the low-pass filtered samples $y\left( n\Delta \right)$, $n = 0, \ldots, N-1$ are a sufficient characterization of the sparse signal $s\l t \r$ in (\ref{FRI}).
\label{thm:BLS}
\end{theorem}
\begin{proof} To show that this statement holds, we start with the observation that modulating the low-pass samples results in the Fourier series of the sparse signal in (\ref{ht}). More precisely, 
\begin{equation}
\label{gn}
{g_n} = {\frac{y\left( {n\Delta } \right)}{\sqrt{\Delta}}{e^{\jmath {\sf{Q}}\left( {n\Delta } \right)}}} \DEq{BLS}  \sum\limits_{\left| m \right| \leqslant {f_c = M}} {\widehat h\left[ m \right]{e^{\jmath \frac{{{\omega _0}m}}{b}n\Delta }}}.
\end{equation}
Also, from (\ref{SOCE}), the Fourier coefficients $\widehat h\left[ m \right]$ are a linear combination of complex exponentials. In vector-matrix notation, we have, $\nvec{g} = \bf{V}\widehat{\nvec{h}}$ (cf.~(\ref{top1})). From (\ref{gn}), we estimate $\widehat h\left[ m \right]$ using the inverse Fourier transform, that is, $\widehat{\nvec{h}} = \bf{V}^+ \nvec{g}$ where $\l \cdot\r^+$ denotes the matrix pseudo-inverse. A unique solution to this system of equations exists provided that

\begin{chequation}
\begin{equation}
{N\geq 2f_c + 1}, \quad {f_c} = \left\lfloor {T/2\Delta } \right\rfloor.
\end{equation}
\end{chequation}

Having computed $\widehat{\mathbf{h}}$, we are now left with the task of {estimating} $\left\{ {{c'_k},{t'_k}} \right\}_{k = 0}^{K - 1}$ associated with the sparse signal in (\ref{ht}). {In spectral estimation theory \cite{Stoica:1997}, it is well known that the sum of complex exponentials in (\ref{SOCE}) admits an autoregressive form which allows us to write, 
\begin{equation}
\label{AP}
\widehat h\left[ m \right] + \sum\limits_{k = 1}^K {r\left[ k \right]\widehat h\left[ {m - k} \right] = 0}. 
\end{equation}
The $K$--tap filter defined by $\{r\left[ k \right]\}$ is known as the \emph{annhilating filter} \cite{Stoica:1997,Vetterli:2002}} which is used to estimate the non-linear unknowns $\left\{ {{t'_k}} \right\}_{k = 0}^{K - 1}$ provided that $\{t'_k\}_{k=0}^{K-1}$ are distinct and $\widehat h\left[ m \right], m \in [-K,K]$ is known, thus implying,
\begin{chequation}
\begin{equation}
{f_c} \geqslant K.
\end{equation}
\end{chequation}
By combining conditions (C1) and (C2), we finally obtain, 
\begin{equation}
\label{FRIcon}
N \geqslant T/\Delta  + 1.
\end{equation}

Whenever (\ref{FRIcon}) holds, a recovery procedure from the FRI literature \cite{Vetterli:2002,Blu:2008} can be directly applied. To this end, (\ref{FRIcon}) guarantees that we can estimate the filter $r$ in (\ref{AP}) which is then used for constructing a polynomial of degree $K$, $$\mathsf{R}\left( z \right) = \sum\limits_{k = 0}^K {r\left[ k \right]{z^{ - k}}}.$$ The $K$--roots of this polynomial, that is, ${u_k} = {e^{-\jmath \omega_0 {t_k}/{b}}}$, encode the information about $\left\{ {{t_k}} \right\}_{k = 0}^{K - 1}$. Let ${\widetilde t_k}$ denote the estimate of $t_k$. Then, by factorizing $\mathsf{R}\left( z \right)$, we estimate the roots $\widetilde u_k$ which is used to estimate ${\widetilde t_k} = \left( {b/{\omega _0}} \right)\angle { \widetilde u_k}$. To determine $\left\{ {{c_k}} \right\}_{k = 0}^{K - 1}$ in (\ref{FRI}), we first construct the quadratic polynomial ${\sf{Q}}\left( \widetilde t_k \right) = \left( {a{\widetilde t_k^2} + 2p \widetilde t_k} \right)/2b$. There on, we estimate $\widetilde c_k$ by solving the least--squares problem since $c_k$'s in \eqref{SOCE} linearly depend on   known quantities.
\end{proof}

{
\subsection{Generalization to Arbitrary Bandlimited Sampling Kernels}
\label{sec:GABSK}

For the same recovery condition (cf.~(\ref{FRIcon})), our result straight-forwardly generalizes to any arbitrary, bandlimited sampling kernel of the form, 
\begin{equation}
\label{FBK}
\psi \in \mathcal{B}^{\Omega_m}_{\bsl_{\sf FT}}, \ {\varphi _{{\sf{BL}}}}\left( t \right) = \frac{1}{{\sqrt \Delta  {K_b}}}{m^*_{\bsls}}\left( t \right){e^{ - \jmath \frac{p}{b}t}} \psi \l \Delta^{-1}t \r
\end{equation}%
provided that $\psi \in \mathcal{B}^{\Omega_m}_{\bsl_{\sf FT}}$ and the Fourier transform of $\psi$ does not vanish in the interval $[-\Omega_m,\Omega_m]$. This is a consequence of a generalized version of Proposition~\ref{prp:BLS} (cf.~Section~\ref{prf:prp:BLS}). {Arbitrary bandlimited kernels} result in a version of (\ref{gn}),
\[
{g_n} = \sum\limits_{\left| m \right| \leqslant {f_c}} {\widehat h\left[ m \right]\widehat \psi \left( {\frac{\omega _0m\Delta}{b}} \right){e^{\jmath \frac{{{\omega _0}m}}{b}n\Delta }}}  \Leftrightarrow \nvec{g} = \bf{V D} \widehat{ \nvec{h}},
\]
where $\bf{D}$ is a diagonal matrix composed of Fourier series coefficients of $\psi$, that is, ${\{{\widehat \psi \left( {{\Delta \omega _0}m/b} \right)}\}_{\left| m \right| \leqslant {f_c}}}$. Given $\psi$, $\widehat{\nvec{h}}$ can be ``deconvolved'' using $\widehat{\nvec{h}} = \bf{D}^{-1} \bf{V}^+ \nvec{g}$ (cf.~Section~\ref{sec:int}).}

\subsection{Special cases of the SAFT domain result}

By appropriately selecting the parameter matrix $\bsls$, we can directly derive results for any of the operations described in Table~\ref{tab:1}. {Next, we revisit some examples in the literature which are special cases of the SAFT domain}. \\

{
\noindent{\bf Sparse or FRI Sampling in Fourier Domain}
 
\noindent The FRI sampling result which was derived in context of Fourier domain \cite{Vetterli:2002} is a special case of Theorem~\ref{thm:BLS}. By setting, 
\begin{equation*}
\bsls = \left[ {\begin{array}{*{20}{c}}
  0 & 1 &\vline & 0 \\ 
  -1 &0 &\vline & 0 
\end{array}} \right] \equiv  \pmb\Lambda_\mathsf{FT},
\end{equation*}
we note that $\mathsf{Q}\l t \r = 0$ and $s\l t \r = h\l t \r$ in \eqref{HT}. In this case {the} sampling rate is $\Delta = \pi/\Omega_m$ and provided that $N\geq T/\Delta + 1$, the sparse signal can be recovered from low-pass projections in the Fourier domain.} \\

{
\noindent{\bf Sparse Sampling in Fractional Fourier Domain}
 
\noindent Sparse sampling in context of the fractional Fourier domain was discussed in \cite{Bhandari:2010}. As above, this is a special case of Theorem~\ref{thm:BLS}. By setting, 
\begin{equation*}
\bsls = \left[ {\begin{array}{*{20}{c}}
  \cos\theta &   \sin\theta &\vline & 0 \\ 
  -\sin\theta &   \cos\theta &\vline & 0 
\end{array}} \right] \equiv  \pmb\Lambda_{\pmb{\theta}},
\end{equation*}
we note that $\mathsf{Q}\l t \r = \tfrac{t^2}{2}\cot\theta$ and as shown in \cite{Bhandari:2010}, $s\l t \r$ is given by,
\[
s\left( t \right) \DEq{HT} {e^{ - \jmath \frac{{{t^2}}}{2}\cot \theta }}\sum\limits_{n \in \mathbb{Z}} {\sum\limits_{k = 0}^{K - 1} {{c'_k}u_k^n} {e^{\jmath \frac{2\pi}{T}nt}}}, 
\]
where ${c'_k} = {c_k}{e^{\jmath \frac{{t_k^2}}{2}\cot \theta }}$. In this case {the} sampling rate is $\Delta = \pi\sin\theta/\Omega_m$ and provided that $N\geq T/\Delta + 1$, the sparse signal can be recovered from low-pass projections in the Fractional Fourier domain.}\\

\noindent{\bf Link with Super-resolution via Convex Programming} 

\noindent Our work is directly related to recent results on super-resolution based on convex-programming \cite{Candes:2013}. Let,
\[
{\left\| s \right\|_{{\sf{TV}}}} = \sum\limits_{k = 0}^{K - 1} {\left| {{c_k}} \right|}
\]
denote the ${\sf{ TV}}$--norm of the sparse signal $s$ (\ref{FRI}). Also note that, ${\left\| s \right\|_{{\sf{TV}}}} = {\left\| h \right\|_{{\sf{TV}}}}$. Given sampled measurements $\left\{ {y\left( {n\Delta } \right)} \right\}_{n = 0}^{N - 1}$ (\ref{BLS}) we first obtain $\left\{ {{g_n}} \right\}_{n = 0}^{N - 1}$ from (\ref{gn}). We then estimate $\widehat{\nvec{h}} = \bf{V}^+ \nvec{g}$. As a result, we may recast our sparse recovery problem as a minimization problem of the form, 
\begin{equation}
\label{tvmin}
\mathop {\min }\limits_h {\left\| h \right\|_{{\sf{TV}}}} \mbox{ with }{\left\{ {\widehat h\left[ m \right] = \int_0^T {h\left( t \right){e^{\jmath {\omega _0}mt}}dt} } \right\}_{\left| m \right| \leqslant {f_c} }}.
\end{equation}

Theorem~\ref{thm:BLS} assumes knowledge of $K$---a proxy for {\emph{sparsity}} or the {\emph{rate of innovation}}. Accordingly, the sampling criterion (\ref{FRIcon}) is based on a counting principle: the kernel bandwidth should be at least equal to the number of unknowns (cf.~(C2)). With $K$ known, in the absence of perturbations, the unknowns $\{t_k\}_{k=0}^{K-1}$ can be arbitrarily close. In contrast, (\ref{tvmin}) avoids any assumptions on $K$ but relies on a minimum separation principle \cite{Candes:2013}. 

Let ${\mathbb{T}_K} = \left\{ {{t_k}} \right\}_{k = 0}^{K - 1}$ denote the support of $s$ and let us define the minimum separation between any two entries of $\mathbb{T}_K$ by, 
\[\delta_{\rm{min}} \l {\mathbb{T}_K} \r \DE {\inf _{{\mathbb{T}_K}:{t_k} \ne {t_\ell }}}\left| {{t_k} - {t_\ell }} \right|.\]
We can now repurpose our generalized result in {the} context of super-resolution \cite{Candes:2013}. The formal result is as follows.

\begin{theorem}[Exact Recovery based on Minimum Separation Principle] Let ${\mathbb{T}_K} = \left\{ {{t_k}} \right\}_{k = 0}^{K - 1}$ be the support set of the sparse signal $s$ and ${f_c} = \left\lfloor {T/2\Delta } \right\rfloor$ be the cut-off frequency of the sampling kernel defined in (\ref{BLK}). If ${\delta _{{\text{min}}}}\left( {{\mathbb{T}_K}} \right){f_c} \geqslant 2$ then $h$ (and hence $s$) is a unique solution to (\ref{tvmin}).
\end{theorem}
The proof of this theorem follows from \cite{Candes:2013}. Even though our discussion is quite general (due to Table~\ref{tab:1}), extension to \cite{Candes:2013} comes at no extra cost---the computational complexity and the exact recovery principle remain the same. A similar parallel can be drawn with the work on atomic norms \cite{Bhaskar:2013}.

\subsection{Sparse Signals and Smooth Time-limited Kernels}
\label{subsec:TLSS}

In the previous section, we focused on sampling kernels which were bandlimited (\ref{FBK}). However, in applications, {the sampling kernels may be pulses or echoes} that are time-limited \cite{Hernandez-Marin:2007,Bhandari:2014,Bhandari:2016a,Bhandari:2016b,Tur:2011,Bhandari:2015c}. {In this case}, we can model such a sampling kernel as a SAFS (\ref{safs}), 
\begin{align}
\label{TLK}
\psi\left( t \right) & \DEq{safs} \sum\limits_{n\in\mathbb{Z}} {\widehat \psi_{\bsls}\left[ n \right]\kappa_{\bslsi}^* \left( {n{\omega _0},t} \right)}, \quad \psi\l t \r =0, t \notin [0,T] \notag \\ 
			& \ = \ \sum\limits_{n\in\mathbb{Z}} {\widehat \psi_{\bsls}\left[ n \right]
			{e^{ - \jmath {\sf{Q}}\left( t \right)}} {\PH}\left( {n{\omega _0}} \right){e^{\jmath \frac{{n{\omega _0}t}}{b}}}},
\end{align} 
where $\PH \left( \omega  \right) = {e^{\jmath \frac{\omega }{b}\left( {dp - bq} \right)}}{e^{-\jmath \frac{{d{\omega ^2}}}{{2b}}}}$ and ${\sf{Q}}\l t \r$ is defined in (\ref{Qt}). As usual, the {filtered spikes} are given by $y\l t \r = \l s \SC \psi\r \left( t \right)$. A further simplification of the measurements is due to the following proposition which deals with time-limited kernels. 
\begin{proposition}[Time-limited Case] Let $s$ be the sparse signal defined in (\ref{FRI}) and the time-limited sampling kernel $\psi\left( t \right)$ be defined in (\ref{TLK}). Then, the {filtered spikes} simplify to
\begin{equation}
\label{TLS}
y\l t \r = e^{-\jmath\sf{Q}\l t \r} \sum\limits_{m\in\mathbb{Z}} \widehat{h}[m]
{\widehat \psi}_{\bsls} \left[ m \right] \PH \left( {m{\omega _0}} \right)
{e^{\jmath \frac{{{\omega _0}m}}{b}t}},
\end{equation}
where $\sf{Q}$ is defined in (\ref{Qt}) and $\widehat{h}$ is defined in (\ref{SOCE}).
\label{prp:TLS}
\end{proposition}
This result is proved in Section~\ref{prf:prp:TLS}. In defining $\psi$ in (\ref{TLK}), time-limitedness was the only assumed property for developing the SAFS representation. In practice one would expect $\psi$ to be a bounded, smooth function. From Fourier regularity conditions, we know that a function is bounded and $N$ times continuously differentiable provided that \cite{Eldar:2015}, 
\[\int {|\widehat \psi_{\bsl_{\sf FT}}\left( \omega  \right)|\left( {1 + {{\left| \omega  \right|}^N}} \right)d\omega }  < \infty.\]
In many cases, $\psi$ is a smooth and compactly supported pulse/function \cite{Hernandez-Marin:2007,Bhandari:2016a,Bhandari:2016b}. %Although never discussed in literature, 
The next proposition states a sufficient condition for smoothness (or differentiability) of a function $\psi$ in the SAFT domain. 
\begin{proposition}
\label{prp:decay}
(Smoothness and Decay) A function $f$ is bounded and $N$ times continuously differentiable with bounded derivatives provided that (ignoring constant $\K$)
\[\int {\left| {{{\left[ {{\sf L}\left( \omega  \right) + a\mathscr{D}} \right]}^N}\widehat f_{\bsls}\left( \omega  \right)} \right|d\omega }  < \infty,\]
where $\mathscr{D}$ is the usual derivative operator and ${\sf L}$ is some linear polynomial. The action of ${{\left[ {{\sf L}\left( \omega  \right) + a\mathscr{D}} \right]}^N}$ on $\widehat{f}_{\bsls}$ should be understood in the operator sense. 
\end{proposition}
The proof of this proposition is based on recurrence relations and is presented in Section \ref{prf:prp:decay}. As an example, consider $N=1$. In this case (by using (\ref{Zk}) and (\ref{zkrec})), $\mathscr{T}_{\textrm{SAFT}} : f^{\l 1 \r}\l t \r \to{\sf{L}}\left( \omega  \right)\widehat f_{\bsls}\left( \omega  \right) + a\mathscr{D}\widehat f_{\bsls}\left( \omega  \right)$. In view of this result, 
\begin{align*}
  |{\mathscr{D} f}\left( t \right)| & \leq \int {\left| {{{\sf L}}\left( \omega  \right){\widehat f}_{\bsls}\left( \omega  \right) + a \mathscr{D} \widehat f_{\bsls}\left( \omega  \right)} \right|d\omega }  \hfill \\
   & \leq 
|| {\sf L}\widehat f_{\bsls}|{|_{{L_1}}} + \left| a \right|||\mathscr{D}\widehat f_{\bsls}||_{L_1} < \infty,
\end{align*}
where the last result is due to Minkowski's inequality. Note that for $N=1$, both $\widehat{f}_{\bsls}$ and $\mathscr{D}\widehat{f}_{\bsls}$ should be in $L_1$. For the case of {the} Fourier transform, we have $a=0$ (cf.~Table~\ref{tab:1}) and the result collapses to ${\left| \omega  \right|^N}|\widehat f_{\bsl_{\sf FT}}\left( \omega  \right)|$. 

{The} smoothness properties of kernels are of significant interest in the context of sampling and approximation theory \cite{Blu:1999}. While a detailed discussion is beyond the scope of this work, for sparse sampling,  Proposition~\ref{prp:decay} is enough to establish that $\widehat{\psi}_{\bsls}$ decays to zero whenever $\psi$ is a smooth kernel. Consequently, for some $f_c = M>0$, we have $\widehat{\psi}_{\bsls}[m] = 0, |m| > M$ and therefore, $M$ dictates the recovery bound for perfect reconstruction of $s$. Without loss of generality, by setting $\Delta = 1$ and,
\[\widehat{\psi}'_{\bsls} \left[ m \right] \DE \widehat{\psi}_{\bsls} \left[ m \right]\PH \left( {m{\omega _0}} \right),\]
the samples take the form of
$$g_n \DEq{TLS} \sum\limits_{|m| \leq M} \widehat{h}[m]
{\widehat \psi}'_{\bsls} \left[ m \right] 
{e^{\jmath \frac{{{\omega _0}m}}{b}n}}, n= 0,\ldots,N-1.$$
This problem is similar to the one in {Section~\ref{sec:GABSK}} and in view of Theorem~\ref{thm:BLS}, perfect recovery is guaranteed provided that $N\geq2M+1$ with $M\geq K$. 

\subsection{Sparse Signals and the Gabor Transform Kernel}
{Several recent works consider} recovery of sparse signals from Gabor transform measurements. Aubel \etal \cite{Aubel:2017} study this problem in the context of super-resolution. Similarly, {Matusiak \etal \cite{Matusiak:2012} studied the problem of sparse sampling using Gabor frames.} Recently, Eldar \etal \cite{Eldar:2015a} developed algorithms for recovery of sparse signals in the context of the phase retrieval. {Uniqueness guarantees with respect to phase retrieval problem for the Gabor transform were reported by Jaganathan \etal in \cite{Jaganathan:2016}.} In all of these cases, the results were developed for the Fourier domain ($\bsls = \bsl_{\sf{FT}}$). Here, we generalize the sparse recovery problem to the SAFT domain. For this purpose, we introduce the Gabor transform associated with the SAFT domain together with some basic mathematical properties. % and study recovery of sparse signals given its low-pass projections in Gabor transform domain.

\begin{definition}[Gabor Transform for the SAFT Domain] Let $f$ be a function with well defined SAFT and $\psi \in L_2$ be some window. We define the SAFT Gabor transform (SAFT-GT) by,
\begin{equation}
\label{SAFTGT}
\GT{f}{\psi}{\tau}{\omega} = \int {f\left( t \right)\psi\left( {t - \tau } \right)\kappa _{\bsls}^*\left( {t,\omega } \right)dt}.
\end{equation}
\end{definition}
Next, we {derive} the inversion formula linked with the SAFT domain. Without loss of generality, we assume that $\K = 1$.
\begin{proposition}[Inversion Formula for the SAFT-GT]%
\label{prp:IGT}
Let $f$ be some function with a well defined SAFT and $\psi_1,\psi_2 \in L_2$ be window functions. Furthermore, let $\GT{f}{\psi_1}{\tau}{\omega}$ denote the SAFT-GT of $f$ defined in (\ref{SAFTGT}). Provided that $b\left\langle {{\psi _1},{\psi _2}} \right\rangle  = 1$, the inverse SAFT-GT is defined by, 
\begin{equation}
\label{SAFTIGT}
f\left( t \right) = \iint_{\tau ,\omega } {\GT{f}{\psi_1}{\tau}{\omega}\psi_2\l t - \tau \r{\kappa ^*_{\bslsi}}\left( {\omega ,t} \right)d\omega d\tau }.
\end{equation}
\end{proposition}
A proof of this proposition is presented in Section~\ref{prf:prp:IGT}. 

In the context of the SAFT-GT, the sparse signal (\ref{FRI}), or alternatively (\ref{hm}), can be represented as, 
%
%Recovery of sparse signals from Gabor transform measurements,
\begin{align}
y\left( {\tau ,{\omega}} \right) & \DEq{SAFTGT} \underbrace{\PH^* \left( \omega  \right)\sum\limits_{k = 0}^{K - 1} {{c'_k}\psi \left( {{t_k} - \tau } \right){e^{-\jmath \frac{\omega }{b}{t_k}}}}}_{\GT{s}{\psi}{\tau}{\omega}}.
\label{gty}
\end{align}
{Let $y'\left( {t ,\omega } \right) = y\left( {t,\omega } \right)\Phi \left( \omega  \right)$ where $\PH$ is a modulation operation with $\PH \PH^* = 1$.}
\begin{equation}
\label{yt}
y' \l t,\omega \r \DE \sum\limits_{k = 0}^{K - 1} {{c'_k}\psi \left( {{t_k} - t } \right){e^{-\jmath \frac{\omega }{b}{t_k}}}}.
\end{equation}
Since $\psi$ is a known, smooth window, a finite Fourier series approximation---time scaled by $b$---suffices to approximate $\psi$, 
\[\widetilde \psi_M \left( t \right) \approx \sum\limits_{\left| m \right| \leqslant M} {\widehat\psi_{\bsl_{\sf FT}} \left[ m \right]{\kappa_{\bsl_{\sf FT}} ^*}\left( {t,{\omega _0}m/b} \right)}, \quad \omega_0 = 2\pi b/T. \]
Let $M = f_c$ and assume sampled measurements of the form, 
\[\left. {y\left( {\tau ,\omega } \right)} \right|_{\tau  = n{\Delta }}^{\omega  = \ell{\omega_0 }}, \ \ \left( {\tau ,\omega } \right) \in \left[ {0,N-1} \right] \times \left[ { - {f_c},{f_c}} \right].
\]
Next, we study the recovery of $s$ in two separate cases:
\begin{enumerate}[leftmargin=12pt,label={$\blacksquare$},itemsep = 8pt]

\item Case 1: Recovery with a fixed $\omega$. 

\noindent By fixing $\omega = \omega_0$ and $\psi \to \psi_M$,
\[y\left( {n\Delta ,{\omega _0}} \right) = \underbrace{{\PH ^*}\left( {{\omega _0}} \right)}_{\sf{Constant}}\sum\limits_{k = 0}^{K - 1} {\underbrace {{c'_k}{e^{-\jmath \frac{{{\omega _0}}}{b}{t_k}}}}_{{\app{c}_k}}\psi_M \left( {{t_k} - n\Delta } \right)} \]

where the weights reflect the effect of the SAFT, $\app{c}_k \DE c'_k{e^{-\jmath \frac{{{\omega _0}}}{b}{t_k}}} = {c_k}{e^{\jmath {\sf{Q}}\left( {{t_k}} \right)}}{e^{-\jmath \frac{{{\omega _0}}}{b}{t_k}}}.$ Finally, note that this is the classical sparse sampling problem (cf.~compare (\ref{sparsemodel})), 
\[y'\left( {n\Delta ,{\omega _0}} \right) = \sum\limits_{k = 0}^{K - 1} {\app{c}_k{{\overline \psi  }_M}\left( {n\Delta  - {t_k}} \right)}
 \]
where ${\overline \psi  _M} = {\psi _M}\left( { -  \cdot } \right)$ and {the result of Section~\ref{sec:GABSK}} can be directly extended to the case at hand. Provided that $f_c = M\geq K$ and $N\geq 2M+1$, perfect recovery is guaranteed. 

\item Case 2: Recovery using both $\l \tau,\omega \r$.

From (\ref{yt}) and $\psi_M \approx \psi$, we have, 
\begin{align}
y'\left( {n\Delta ,\ell {\omega _0}} \right)&  = \sum\limits_{k = 0}^{K - 1} {{c_k}\sum\limits_{\left| m \right| \leqslant M} { {{\widehat\psi_{\bsl_{\sf FT}}}} \left[m \right]{e^{\jmath \frac{{{t_k}}}{b}\left( {m{\omega _0} - \ell {\omega _0}} \right)}}{e^{ - \jmath \frac{{m{\omega _0}}}{b}n\Delta }}} } \notag\\
& = \sum\limits_{\left| m \right| \leqslant M} {{\app{y}_{m,\ell }}{e^{ - \jmath \frac{{m{\omega _0}}}{b}n\Delta }}}
\label{STSR}
\end{align}
where, ${\app{y}_{m,\ell }} \DE \sum\limits_{k = 0}^{K - 1} {{\widehat\psi_{\bsl_{\sf FT}}}} \left[ m \right]{c_k}{e^{\jmath \frac{{{t_k}}}{b}\left( {m{\omega _0} - \ell {\omega _0}} \right)}}.$
This development is the standard form (with regards to the Fourier domain) and the results in \cite{Aubel:2017} can be extended to solve for (\ref{STSR}). For a Gaussian window function, the exact recovery principle based on minimum separation condition is discussed in \cite{Aubel:2017} (cf.~Theorem 11) and applies to (\ref{STSR}). 
\end	{enumerate}

\section{Conclusion}
\label{sec:conclusion}
In this work, we considered the recovery of sparse signals from their low-pass/bandlimited projections in the SAFT domain. Since the SAFT parametrically generalizes a number of interesting transformations listed in Table~\ref{tab:1}, our work presents a unifying approach to the problem of sampling and recovery of sparse signals. Starting with a review of Shannon's sampling theorem for signals that are bandlimited in the SAFT domain, we developed conditions for exact recovery of sparse signals when, (1) sampling with arbitrary, bandlimited kernels, (2) sampling with smooth, time-limited kernels and, (3) recovery from Gabor transform measurements linked with the SAFT domain. By setting $\bsls = \bsl_{\sf{FT}}$ and $\bsl_{\theta}$, our results coincide with previously discussed results linked with the Fourier and Fractional Fourier domain, respectively.

%\appendices

\section{Appendix: Auxiliary Proofs and Computations}
\label{sec:app}

\subsection{Proof of Proposition~\ref{prp:BLS}}
\label{prf:prp:BLS} 

Here, we will prove a more general result. Let $\psi$ be an arbitrary function with a well defined Fourier transform and let $\varphi\l t  \r = \l {\Delta  {K_b^2}} \r^{-1/2}{m^*_{\bsls}}\left( t \right) {e^{ - \jmath \frac{p}{b}t}} \psi \l t \r$.
By using the definition of the SAFT-convolution (\ref{sconv}), we have
\begin{align}
y\l t \r & = \K m_{\bsls}^*\l t \r   \l \upo{s}*\upo{\varphi} \r \l t\r  \notag \\ 
&  \DA{a}\Delta_c m_{\bsls}^*\l t \r         \l e^{\jmath\tfrac{p}{b} t}h\l t \r * e^{\jmath\tfrac{p}{b} t}\psi\l t \r  \r \notag \\
& \DA{b} \Delta_c { m_{\bsls}^*\l t \r}   e^{\jmath\tfrac{p}{b} t}      \l h\l t \r * \psi\l t \r  \r \notag \\
& \DA{c} \Delta_c{e^{-\jmath\sf{Q}\l t \r}} \sum\limits_{m\in\mathbb{Z}} \widehat{h}[m]\widehat \psi\left( {\frac{{{\omega _0}m}}{b}} \right){e^{\jmath \frac{{{\omega _0}m}}{b}t}},
\label{dcv}
\end{align}
where $\Delta_c = 1/\sqrt\Delta$, ${\rm{(a)}}$ is due to (\ref{sconv}), ${\rm{(b)}}$ is due to invariance of complex exponentials under convolution operation (eigen-function property) and ${\rm{(c)}}$ is because $h$ is a $T$-periodic function (\ref{ht}). Here, $y\l t \r$ is completely characterized by the Fourier series coefficients of $h$ and $\psi$. With $\psi\l t \r = \sincD{t}$, we have, $\widehat \psi\left( \omega  \right) =\Delta {\ind_{\left[ { - \frac{\pi }{\Delta },\frac{\pi }{\Delta }} \right]}}\left( \omega  \right)$ where  $\omega  = \omega_0 m/b$ and since $\widehat \psi\left(m \omega_0 \Delta / b  \right) = 0, m >T/2\Delta$, (\ref{BLS}) holds.

\subsection{Proof of Proposition~\ref{prp:TLS}}
\label{prf:prp:TLS}

We will start by developing $\upo{s}*\upo{\psi}$ which appears in the definition of the convolution operator (\ref{sconv}). Note that, 
\begin{align*}
\up{s}{t}  & =  {e^{ - \jmath \frac{p}{b}t}}h\left( t \right), \quad \mbox{ and } \\
\up{\psi}{t}    & ={e^{ - \jmath \frac{p}{b}t}}\sum\limits_{n \in \mathbb{Z}} {{{\widehat \psi}_{\bsls} \left[ n \right]  {\PH}\left( {n{\omega _0}} \right)} {e^{\jmath \frac{{n{\omega _0}}}{b}t}}}.
\end{align*}
By letting $\widehat{z}\left[ n \right] \DE {\widehat \psi}_{\bsls} \left[ n \right] \Phi \left( {n{\omega _0}} \right)$, we may re-write $\up{\psi}{t}$ in terms of Fourier series of $z\l t \r$, 
\[
\up{\psi}{t}  ={e^{ - \jmath \frac{p}{b}t}}\sum\limits_{n \in \mathbb{Z}} \widehat{z}\left[ n \right] {e^{\jmath \frac{{n{\omega _0}}}{b}t}} \equiv {e^{ - \jmath \frac{p}{b}t}} z\l t \r.
\]
Based, on this, we now develop, 
\begin{align*}
( \upo{s}*\upo{\psi} ) \l t \r &  =  {e^{ - \jmath \frac{p}{b}t}}h\left( t \right)*{e^{ - \jmath \frac{p}{b}t}}z\left( t \right) \hfill \\
&   \DA{a} {e^{ - \jmath \frac{p}{b}t}}\left( {h*z} \right)\left( t \right) \hfill \\
&   \DA{b}  {e^{ - \jmath \frac{p}{b}t}} T \sum\limits_{m \in \mathbb{Z}} {\widehat h\left[ m \right]\widehat z\left[ m \right]{e^{\jmath \frac{{m{\omega _0}}}{b}t}}},
\end{align*}
where ${\rm{(a)}}$ is due to invariance of complex exponentials under convolution operation (eigen-function property) and ${\rm{(b)}}$ is due to convolution-product theorem for the Fourier series. For simplicity, let us assume that $T\K = 1$. Then, we have, 
\begin{align}
y\l t \r & = \K m^*_{\bsls}\l t\r {e^{ - \jmath \frac{p}{b}t}} T \sum\limits_{m \in \mathbb{Z}} {\widehat h\left[ m \right]\widehat z\left[ m \right]{e^{\jmath \frac{{m{\omega _0}}}{b}t}}}, \notag \\
& = e^{-\jmath\sf{Q}\l t \r} \sum\limits_{m\in\mathbb{Z}} \widehat{h}[m]
{\widehat \psi}_{\bsls} \left[ m \right] \Phi \left( {m{\omega _0}} \right)
{e^{\jmath \frac{{{\omega _0}m}}{b}t}},
\end{align}
{which} completes our proof.

\subsection{Proof of Proposition~\ref{prp:decay}}
\label{prf:prp:decay}

In order to prove this result, {we begin with the observation}, 
\begin{equation}
\label{FW}
\left| {{f^{\left( k \right)}}\left( t \right)} \right| = \left| {\int {{F_k}\left( \omega  \right)\kappa _{\bslsi}^*\left( {\omega ,t} \right)d\omega } } \right| \leqslant \int {\left| {{F_k}\left( \omega  \right)} \right|d\omega }  < \infty, 
\end{equation}
where ${{F_k}\left( \omega  \right)}$, expressed as a function of $\saft{f}$, is the SAFT of ${{f^{\left( k \right)}}\left( t \right)}$ which is yet to be determined. In analogy to the Fourier transform, $F_k\l \omega \r \propto \l \jmath \omega \r^k \widehat{f}_{\bsl_{\sf{FT}}}\l \omega \r$. To set up this proof, we will start with defining smooth functions. Let $f$ be some function with norm defined as $\forall m,n \in {\mathbb{Z}_ + } \cup \left\{ 0 \right\},{\left\| f \right\|_{m,n}} = {\sup _{t \in \mathbb{R}}}\left| {{t^m}{f^{\left( n \right)}}\left( t \right)} \right|$. Then, we say $f$ is smooth or $f \in {\sf{S}}$ provided that ${\left\| f \right\|_{m,n}} < \infty$. For functions $u$ and $v$ bounded in this norm, integration by parts results in, 
\begin{equation}
\label{IBP}
\int {{u^{\left( 1 \right)}}\left( t \right)v\left( t \right)dt}  =  - \int {u\left( t \right){v^{\left( 1 \right)}}\left( t \right)dt}.
\end{equation}
Let $u\left( t \right) \DE {f^{\left( k \right)}}\left( t \right)$ and $v\left( {\nvec{r}} \right) \DE \kappa _{\bsls}^*\left( \nvec{r} \right)$ where $\nvec{r} = \left[ t \ \ \omega \right]^{\top}$. Also note {two useful} relations that will be used shortly, 
\begin{align}
  {\partial _t}v\left( \nvec{r} \right) &= \beta \left( {p - \omega  + at} \right) v\left( \nvec{r} \right) \hfill \\
  {\partial _\omega }v\left( \nvec{r} \right) & = \beta\left( {\left( {d\omega  - \mu } \right) - t} \right)v\left( \nvec{r} \right), \ \  \mu \DE dp-bq, 
\label{vw}
\end{align} 
where $\beta = \jmath/b$. Next, with $Z_0\l \omega \r \DE \widehat{f}_{\bsls}\l \omega \r$, let us define a sequence of functions $\{Z_k\}, k \in \mathbb{Z}_+ \cup \{0\}$,
\begin{equation}
\label{Zk}
{Z_k}\left( \omega  \right) \DE \int {{f^{\left( k \right)}}\left( t \right)v\left( {t,\omega } \right)dt}  \equiv \int {{f^{\left( k \right)}}\left( t \right){\kappa ^*_{\bsls}}\left( {t,\omega } \right)dt}.
\end{equation}
Similarly, we also define $\{X_k\}, k \in \mathbb{Z}_+ \cup \{0\}$,
\[{X_k}\left( \omega  \right) \DE \int {t{f^{\left( k \right)}}\left( t \right)v\left( {t,\omega } \right)dt}.\]
Thanks to the sequences $\left\{ {{Z_k},{X_k}} \right\}$, (\ref{IBP}) can be reduced to the following recursive form, 
\begin{equation}
\label{el1}
%\boxed
{{Z_{k + 1}}\left( \omega  \right) \DEq{IBP} \beta \left( {\omega  - p} \right){Z_k}\left( \omega  \right) - a\beta {X_k}\left( \omega  \right).}
\end{equation}
Our result relies on $Z_k$ and hence, we must eliminate $X_k$. We do so by observing that, 
\begin{align}
  Z_k^{\left( 1 \right)}\left( \omega  \right) & = \int {{f^{\left( k \right)}}\left( t \right)\underbrace {{\partial _\omega }v\left( {t,\omega } \right)}_{{\text{Replace by (\ref{vw})}}}dt}  \hfill \notag \\
   & = \int {{f^{\left( k \right)}}\left( t \right)\beta \left( {\left( {d\omega  - \mu } \right) - t} \right)v\left( {t,\omega } \right)dt}  \hfill \notag \\
  & = {\beta \left( {d\omega  - \mu } \right){Z_k}\left( \omega  \right) - \beta {X_k}\left( \omega  \right)}.
\label{el2}
\end{align}
We now solve for $Z_k$ from the system of equations (\ref{el1}), (\ref{el2}), 
\[\left[ {\begin{array}{*{20}{c}}
  {{Z_{k + 1}}\left( \omega  \right)} \\ 
  {{Z_k}\left( \omega  \right)} 
\end{array}} \right] = \beta \left[ {\begin{array}{*{20}{c}}
  {\omega  - p}&{ - a} \\ 
  {d\omega  - \mu }&1 
\end{array}} \right]\left[ {\begin{array}{*{20}{c}}
  {{Z_k}\left( \omega  \right)} \\ 
  {{X_k}\left( \omega  \right)} 
\end{array}} \right]\]
which leads to a simple recursive equation, 
\begin{align*}
{Z_{k + 1}}\left( \omega  \right) & = {\sf L}\left( \omega  \right){Z_k}\left( \omega  \right) + aZ_k^{\left( 1 \right)}\left( \omega  \right),\\
& = \left[{{\sf L}}\left( \omega  \right) + a\mathscr{D}\right] Z_k\l \omega \r
\end{align*}
where ${{\sf L}}\left( \omega  \right)$ is linear polynomial ${{\sf L}}\left( \omega  \right) = \beta \left( {\omega bc + a\mu  - p} \right)$ with $\mu = dp-bq$ and is completely characterized by $\bsls$. Now since $Z_0\l \omega \r = \widehat{f}_{\bsls}\l \omega \r$, we observe that, 
\begin{equation}
\label{zkrec}
{Z_k}\left( \omega  \right) = {\left[ {{\sf L}\left( \omega  \right) + a\mathscr D} \right]^k}\widehat f_{\bsls}\left( \omega  \right), \quad Z_0\l \omega \r = \widehat{f}_{\bsls}\l \omega \r
\end{equation}
and by definition (\ref{Zk}), $F_k = Z_k$. Back substituting $Z_k$ in (\ref{FW}) leads to the result of Proposition~\ref{prp:decay}.

\subsection{Proof of Proposition~\ref{prp:IGT}}
\label{prf:prp:IGT}

Let us assume (\ref{SAFTIGT}) is true. Furthermore, we have, 
\begin{equation}
\label{roi}
\left\langle {{\kappa ^*_{\bsls}}\left( {x,\omega} \right),{\kappa _{\bslsi}}\left( {\omega ,t} \right)} \right\rangle  = {e^{\jmath \left( {{\sf Q}\left( x \right) - {\sf Q}\left( t \right)} \right)}}\underbrace {\int {{e^{ - \jmath \frac{\omega }{b}\left( {x - t} \right)}}d\omega } }_{b\delta \left( {x - t} \right)}.
\end{equation}
Next, we substitute $\GT{f}{\psi}{\tau}{\omega}$ in (\ref{SAFTIGT}) to obtain,
\[\iint\limits_{\tau,x} {f\left( x \right){\psi _1}\left( {x - \tau } \right){\psi _2}\left( {x - \tau } \right)}\left\langle {{\kappa ^*_{\bsls}}\left( {x,\omega} \right),{\kappa_{\bslsi}}\left( {\omega ,t} \right)} \right\rangle dxd\tau.\]
Thanks to (\ref{roi}), by marginalizing $\omega$ and then $x$, the last equation yields $bf\left( t \right)\left\langle {{\psi _1},{\psi _2}} \right\rangle = I_f$. Setting $I_f = f$ verifies the result.

\ifCLASSOPTIONcaptionsoff
  \newpage
\fi

\bibliographystyle{IEEEtran}

\begin{thebibliography}{10}
\providecommand{\url}[1]{#1}
\csname url@samestyle\endcsname
\providecommand{\newblock}{\relax}
\providecommand{\bibinfo}[2]{#2}
\providecommand{\BIBentrySTDinterwordspacing}{\spaceskip=0pt\relax}
\providecommand{\BIBentryALTinterwordstretchfactor}{4}
\providecommand{\BIBentryALTinterwordspacing}{\spaceskip=\fontdimen2\font plus
\BIBentryALTinterwordstretchfactor\fontdimen3\font minus
  \fontdimen4\font\relax}
\providecommand{\BIBforeignlanguage}[2]{{%
\expandafter\ifx\csname l@#1\endcsname\relax
\typeout{** WARNING: IEEEtran.bst: No hyphenation pattern has been}%
\typeout{** loaded for the language `#1'. Using the pattern for}%
\typeout{** the default language instead.}%
\else
\language=\csname l@#1\endcsname
\fi
#2}}
\providecommand{\BIBdecl}{\relax}
\BIBdecl

\bibitem{Bhandari:2015}
A.~Bhandari, Y.~C. Eldar, and R.~Raskar, ``Super-resolution in phase space,''
  in \emph{{IEEE} Intl. Conf. on Acoustics, Speech and Signal Processing
  (ICASSP)}, Apr. 2015, pp. 4155--4159.

\bibitem{Bhandari:2016}
A.~Bhandari and Y.~C. Eldar, ``A swiss army knife for finite rate of innovation
  sampling theory,'' in \emph{{IEEE} Intl. Conf. on Acoustics, Speech and
  Signal Processing (ICASSP)}, Mar. 2016, pp. 3999 -- 4003.

\bibitem{Donoho:1992}
D.~L. Donoho, ``Superresolution via sparsity constraints,'' \emph{{SIAM}
  Journal on Mathematical Analysis}, vol.~23, no.~5, pp. 1309--1331, Sep. 1992.

\bibitem{Manabe:1992}
T.~Manabe and H.~Takai, ``Superresolution of multipath delay profiles measured
  by {PN} correlation method,'' \emph{{IEEE} Trans. Antennas Propag.}, vol.~40,
  no.~5, pp. 500--509, May 1992.

\bibitem{Candes:2013}
E.~J. Cand{\`{e}}s and C.~Fernandez-Granda, ``Towards a mathematical theory of
  super-resolution,'' \emph{Communications on Pure and Applied Mathematics},
  vol.~67, no.~6, pp. 906--956, Apr. 2013.

\bibitem{Fuchs:1994}
J.~J. Fuchs and H.~Chuberre, ``A deconvolution approach to source
  localization,'' \emph{{IEEE} Trans. Sig. Proc.}, vol.~42, no.~6, pp.
  1462--1470, Jun. 1994.

\bibitem{Zi-qiang:1982}
H.~Zi-qiang and W.~Zhen-dong, ``A new method for high resolution estimation of
  time delay,'' in \emph{{IEEE} Intl. Conf. on Acoustics, Speech and Signal
  Processing (ICASSP)}, vol.~7, May 1982, pp. 420--423.

\bibitem{Pallas:1991}
M.~A. Pallas and G.~Jourdain, ``Active high resolution time delay estimation
  for large {BT} signals,'' \emph{{IEEE} Trans. Sig. Proc.}, vol.~39, no.~4,
  pp. 781--788, Apr. 1991.

\bibitem{Li:1998a}
J.~Li and R.~Wu, ``An efficient algorithm for time delay estimation,''
  \emph{{IEEE} Trans. Sig. Proc.}, vol.~46, no.~8, pp. 2231--2235, Aug. 1998.

\bibitem{Gedalyahu:2010a}
K.~Gedalyahu and Y.~C. Eldar, ``Time-delay estimation from low-rate samples: A
  union of subspaces approach,'' \emph{{IEEE} Trans. Sig. Proc.}, vol.~58,
  no.~6, pp. 3017--3031, Jun. 2010.

\bibitem{Li:2000}
L.~Li and T.~P. Speed, ``Parametric deconvolution of positive spike trains,''
  \emph{Annals of Statistics}, pp. 1279--1301, 2000.

\bibitem{Hernandez-Marin:2007}
S.~Hernandez-Marin, A.~Wallace, and G.~Gibson, ``Bayesian analysis of lidar
  signals with multiple returns,'' \emph{{IEEE} Trans. Pattern Anal. Mach.
  Intell.}, vol.~29, no.~12, pp. 2170--2180, Dec. 2007.

\bibitem{Bhandari:2014}
A.~Bhandari, A.~Kadambi, and R.~Raskar, ``Sparse linear operator identification
  without sparse regularization? {A}pplications to mixed pixel problem in
  time-of-flight/range imaging,'' in \emph{{IEEE} Intl. Conf. on Acoustics,
  Speech and Signal Processing (ICASSP)}, May 2014, pp. 365--369.

\bibitem{Bhandari:2016a}
A.~Bhandari and R.~Raskar, ``Signal processing for time-of-flight imaging
  sensors,'' \emph{{IEEE} Signal Process. Mag.}, vol.~33, no.~4, pp. 2--16,
  Sep. 2016.

\bibitem{Bar-Ilan:2014}
O.~Bar-Ilan and Y.~C. Eldar, ``Sub-nyquist radar via doppler focusing,''
  \emph{{IEEE} Trans. Sig. Proc.}, vol.~62, no.~7, pp. 1796--1811, Apr. 2014.

\bibitem{Rudresh:2017}
S.~Rudresh and C.~S. Seelamantula, ``Finite-rate-of-innovation-sampling-based
  super-resolution radar imaging,'' \emph{{IEEE} Trans. Sig. Proc.}, vol.~65,
  no.~19, pp. 5021--5033, Oct. 2017.

\bibitem{Bhandari:2017a}
A.~Bhandari and T.~Blu, ``{FRI} sampling and time-varying pulses: {S}ome theory
  and four short stories,'' in \emph{{IEEE} Intl. Conf. on Acoustics, Speech
  and Signal Processing (ICASSP)}, Mar. 2017.

\bibitem{Tur:2011}
R.~Tur, Y.~C. Eldar, and Z.~Friedman, ``Innovation rate sampling of pulse
  streams with application to ultrasound imaging,'' \emph{{IEEE} Trans. Sig.
  Proc.}, vol.~59, no.~4, pp. 1827--1842, Apr. 2011.

\bibitem{Burshtein:2016}
A.~Burshtein, M.~Birk, T.~Chernyakova, A.~Eilam, A.~Kempinski, and Y.~C. Eldar,
  ``Sub-nyquist sampling and {F}ourier domain beamforming in volumetric
  ultrasound imaging,'' \emph{{IEEE} Trans. Ultrason., Ferroelectr., Freq.
  Control}, vol.~63, no.~5, pp. 703--716, May 2016.

\bibitem{Unser:2000}
M.~Unser, ``Sampling--50 years after {Shannon},'' \emph{Proc. {IEEE}}, vol.~88,
  no.~4, pp. 569--587, 2000.

\bibitem{Eldar:2015}
Y.~C. Eldar, \emph{Sampling Theory: Beyond Bandlimited Systems}.\hskip 1em plus
  0.5em minus 0.4em\relax Cambridge University Press, 2015.

\bibitem{Vetterli:2002}
M.~Vetterli, P.~Marziliano, and T.~Blu, ``Sampling signals with finite rate of
  innovation,'' \emph{{IEEE} Trans. Sig. Proc.}, vol.~50, no.~6, pp.
  1417--1428, 2002.

\bibitem{Blu:2008}
T.~Blu, P.~L. Dragotti, M.~Vetterli, P.~Marziliano, and L.~Coulot, ``Sparse
  sampling of signal innovations,'' \emph{{IEEE} Signal Process. Mag.},
  vol.~25, no.~2, pp. 31--40, 2008.

\bibitem{Barbotin:2012}
Y.~Barbotin, A.~Hormati, S.~Rangan, and M.~Vetterli, ``Estimation of sparse
  {MIMO} channels with common support,'' \emph{{IEEE} Trans. Commun.}, vol.~60,
  no.~12, pp. 3705--3716, Dec. 2012.

\bibitem{Bhandari:2016b}
A.~Bhandari, A.~M. Wallace, and R.~Raskar, ``Super-resolved time-of-flight
  sensing via {FRI} sampling theory,'' in \emph{{IEEE} Intl. Conf. on
  Acoustics, Speech and Signal Processing (ICASSP)}, Mar. 2016.

\bibitem{Deslauriers-Gauthier:2013}
S.~Deslauriers-Gauthier and P.~Marziliano, ``Sampling signals with a finite
  rate of innovation on the sphere,'' \emph{{IEEE} Trans. Sig. Proc.}, vol.~61,
  no.~18, pp. 4552--4561, 2013.

\bibitem{Chen:2012}
C.~Chen, P.~Marziliano, and A.~C. Kot, ``{2D} finite rate of innovation
  reconstruction method for step edge and polygon signals in the presence of
  noise,'' \emph{{IEEE} Trans. Sig. Proc.}, vol.~60, no.~6, pp. 2851--2859,
  Jun. 2012.

\bibitem{Pan:2014}
H.~Pan, T.~Blu, and P.~L. Dragotti, ``Sampling curves with finite rate of
  innovation,'' \emph{{IEEE} Trans. Sig. Proc.}, vol.~62, no.~2, pp. 458--471,
  Jan. 2014.

\bibitem{Mulleti:2016}
S.~Mulleti and C.~S. Seelamantula, ``Ellipse fitting using the finite rate of
  innovation sampling principle,'' \emph{{IEEE} Trans. Image Proc.}, vol.~25,
  no.~3, pp. 1451--1464, Mar. 2016.

\bibitem{Onativia:2013}
J.~O{\~{n}}ativia, S.~R. Schultz, and P.~L. Dragotti, ``A finite rate of
  innovation algorithm for fast and accurate spike detection from two-photon
  calcium imaging,'' \emph{J. Neural Eng.}, vol.~10, no.~4, p. 046017, Jul.
  2013.

\bibitem{Dogan:2014}
Z.~Dogan, T.~Blu, and D.~V.~D. Ville, ``Detecting spontaneous brain activity in
  functional magnetic resonance imaging using finite rate of innovation,'' in
  \emph{{IEEE} Intl. Symp. on Biomed. Imaging}.\hskip 1em plus 0.5em minus
  0.4em\relax {IEEE}, Apr. 2014.

\bibitem{Seelamantula:2014}
C.~S. Seelamantula and S.~Mulleti, ``Super-resolution reconstruction in
  frequency-domain optical-coherence tomography using the
  finite-rate-of-innovation principle,'' \emph{{IEEE} Trans. Sig. Proc.},
  vol.~62, no.~19, pp. 5020--5029, Oct. 2014.

\bibitem{Pan:2017}
H.~Pan, T.~Blu, and M.~Vetterli, ``Towards generalized {FRI} sampling with an
  application to source resolution in radioastronomy,'' \emph{{IEEE} Trans.
  Sig. Proc.}, vol.~65, no.~4, pp. 821--835, Feb. 2017.

\bibitem{Mulleti:2017}
S.~Mulleti, A.~Singh, V.~P. Brahmkhatri, K.~Chandra, T.~Raza, S.~P. Mukherjee,
  C.~S. Seelamantula, and H.~S. Atreya, ``Super-resolved nuclear magnetic
  resonance spectroscopy,'' \emph{Scientific Reports}, vol.~7, no.~1, Aug.
  2017.

\bibitem{Bhandari:2017b}
A.~Bhandari, F.~Krahmer, and Ramesh, ``On unlimited sampling,'' in \emph{Intl.
  Conf. on Sampling Theory and Applications (SampTA)}, Jul. 2017.

\bibitem{Bhandari:2018a}
A.~Bhandari, F.~Krahmer, and R.~Raskar, ``Unlimited sampling of sparse
  signals,'' in \emph{{IEEE} Intl. Conf. on Acoustics, Speech and Signal
  Processing (ICASSP)}, Apr. 2018.

\bibitem{Murray-Bruce:2017}
J.~Murray-Bruce and P.~L. Dragotti, ``A sampling framework for solving
  physics-driven inverse source problems,'' \emph{{IEEE} Trans. Sig. Proc.},
  vol.~65, no.~24, pp. 6365--6380, Dec. 2017.

\bibitem{Pan:2017b}
H.~Pan, R.~Scheibler, E.~F. Bezzam, I.~Dokmanic, and M.~Vetterli, ``{FRIDA}:
  {F}ri-based doa estimation for arbitrary array layouts,'' in \emph{{IEEE}
  Intl. Conf. on Acoustics, Speech and Signal Processing (ICASSP)}, 2017.

\bibitem{Peleg:1991}
S.~Peleg and B.~Porat, ``Estimation and classification of polynomial-phase
  signals,'' \emph{{IEEE} Trans. Inf. Theory}, vol.~37, no.~2, pp. 422--430,
  Mar. 1991.

\bibitem{Yuan:2012}
X.~Yuan, ``Estimating the {DOA} and the polarization of a polynomial-phase
  signal using a single polarized vector-sensor,'' \emph{{IEEE} Trans. Sig.
  Proc.}, vol.~60, no.~3, pp. 1270--1282, Mar. 2012.

\bibitem{Amar:2010}
A.~Amar, ``Efficient estimation of a narrow-band polynomial phase signal
  impinging on a sensor array,'' \emph{{IEEE} Trans. Sig. Proc.}, vol.~58,
  no.~2, pp. 923--927, Feb. 2010.

\bibitem{Mann:1995}
S.~Mann and S.~Haykin, ``The chirplet transform: physical considerations,''
  \emph{{IEEE} Trans. Sig. Proc.}, vol.~43, no.~11, pp. 2745--2761, 1995.

\bibitem{Baraniuk:1993}
R.~Baraniuk and D.~Jones, ``Shear madness: new orthonormal bases and frames
  using chirp functions,'' \emph{{IEEE} Trans. Sig. Proc.}, vol.~41, no.~12,
  pp. 3543--3549, 1993.

\bibitem{Gribonval:2001}
R.~Gribonval, ``Fast matching pursuit with a multiscale dictionary of
  {G}aussian chirps,'' \emph{{IEEE} Trans. Sig. Proc.}, vol.~49, no.~5, pp.
  994--1001, 2001.

\bibitem{Engen:2011}
G.~Engen and Y.~Larsen, ``Efficient full aperture processing of {TOPS} mode
  data using the moving band chirp $z$-transform,'' \emph{IEEE Trans. Geosci.
  Remote Sensing}, vol.~49, no.~10, pp. 3688--3693, Oct. 2011.

\bibitem{Schock:2004}
S.~Schock, ``A method for estimating the physical and acoustic properties of
  the sea bed using chirp sonar data,'' \emph{{IEEE} Journal of Oceanic
  Engineering}, vol.~29, no.~4, pp. 1200--1217, Oct. 2004.

\bibitem{Martone:2001}
M.~Martone, ``A multicarrier system based on the fractional {Fourier} transform
  for time-frequency-selective channels,'' \emph{{IEEE} Trans. Commun.},
  vol.~49, no.~6, pp. 1011--1020, Jun. 2001.

\bibitem{Harms:2015a}
A.~Harms, W.~U. Bajwa, and R.~Calderbank, ``Identification of linear
  time-varying systems through waveform diversity,'' \emph{{IEEE} Trans. Sig.
  Proc.}, vol.~63, no.~8, pp. 2070--2084, Apr. 2015.

\bibitem{Gori:1981}
F.~Gori, ``Fresnel transform and sampling theorem,'' \emph{Optics
  Communications}, vol.~39, no.~5, pp. 293--297, Nov. 1981.

\bibitem{Chacko:2013}
N.~Chacko, M.~Liebling, and T.~Blu, ``Discretization of continuous convolution
  operators for accurate modeling of wave propagation in digital holography,''
  \emph{Journal of the Optical Society of America A}, vol.~30, no.~10, p. 2012,
  Sep. 2013.

\bibitem{Pei:2007a}
S.-C. Pei and J.-J. Ding, ``Relations between {Gabor} transforms and fractional
  {Fourier} transforms and their applications for signal processing,''
  \emph{{IEEE} Trans. Sig. Proc.}, vol.~55, no.~10, pp. 4839--4850, Oct. 2007.

\bibitem{Yetik:2003}
I.~Yetik and A.~Nehorai, ``Beamforming using the fractional {F}ourier
  transform,'' \emph{{IEEE} Trans. Sig. Proc.}, vol.~51, no.~6, pp. 1663--1668,
  Jun. 2003.

\bibitem{Setala:2010}
T.~Set\"al\"a, T.~Shirai, and A.~T. Friberg, ``Fractional {Fourier} transform
  in temporal ghost imaging with classical light,'' \emph{Physical Review A},
  vol.~82, no.~4, Oct. 2010.

\bibitem{Unnikrishnan:2000}
G.~Unnikrishnan, J.~Joseph, and K.~Singh, ``Optical encryption by double-random
  phase encoding in the fractional {Fourier} domain,'' \emph{Optics Letters},
  vol.~25, no.~12, p. 887, Jun. 2000.

\bibitem{Huang:2011}
Y.~Huang, ``Entropic uncertainty relations in multidimensional position and
  momentum spaces,'' \emph{Physical Review A}, vol.~83, no.~5, May 2011.

\bibitem{Condon:1937}
E.~U. Condon, ``Immersion of the {F}ourier transform in a continuous group of
  functional transformations,'' \emph{Proceedings of the National Academy of
  Sciences}, vol.~23, no.~3, pp. 158--164, Mar. 1937.

\bibitem{Moshinsky:1971}
M.~Moshinsky and C.~Quesne, ``Linear canonical transformations and their
  unitary representations,'' \emph{Journal of Mathematical Physics}, vol.~12,
  no.~8, pp. 1772--1780, 1971.

\bibitem{Healy:2016}
J.~J. Healy, M.~A. Kutay, H.~M. Ozaktas, and J.~T. Sheridan, Eds., \emph{Linear
  Canonical Transforms: Theory and Applications}, ser. Springer Series in
  Optical Sciences.\hskip 1em plus 0.5em minus 0.4em\relax Springer New York,
  2016, vol. 198.

\bibitem{Almeida:1994}
L.~B. Almeida, ``The fractional {F}ourier transform and time-frequency
  representations,'' \emph{{IEEE} Trans. Sig. Proc.}, vol.~42, no.~11, pp.
  3084--3091, 1994.

\bibitem{Tao:2008}
R.~Tao, B.~Deng, W.-Q. Zhang, and Y.~Wang, ``Sampling and sampling rate
  conversion of band limited signals in the fractional {Fourier} transform
  domain,'' \emph{{IEEE} Trans. Sig. Proc.}, vol.~56, no.~1, pp. 158--171, Jan.
  2008.

\bibitem{Bhandari:2012}
A.~Bhandari and A.~I. Zayed, ``Shift-invariant and sampling spaces associated
  with the fractional {F}ourier transform domain,'' \emph{{IEEE} Trans. Sig.
  Proc.}, vol.~60, no.~4, pp. 1627--1637, 2012.

\bibitem{Shi:2012}
J.~Shi, X.~Liu, X.~Sha, and N.~Zhang, ``Sampling and reconstruction of signals
  in function spaces associated with the linear canonical transform,''
  \emph{{IEEE} Trans. Sig. Proc.}, vol.~60, no.~11, pp. 6041--6047, Nov. 2012.

\bibitem{Bhandari:2017}
A.~Bhandari and A.~I. Zayed, ``Shift-invariant and sampling spaces associated
  with the special affine {F}ourier transform,'' \emph{Applied and
  Computational Harmonic Analysis}, Jul. 2017.

\bibitem{Bhandari:2010}
A.~Bhandari and P.~Marziliano, ``Sampling and reconstruction of sparse signals
  in fractional {Fourier} domain,'' \emph{{IEEE} Signal Process. Lett.},
  vol.~17, no.~3, pp. 221--224, 2010.

\bibitem{Bhandari:2018}
A.~Bhandari and A.~I. Zayed, \emph{Frontiers In Orthogonal Polynomials And
  Q-series}.\hskip 1em plus 0.5em minus 0.4em\relax World Scientific, 2018, ch.
  Convolution and Product Theorem for the {Special Affine Fourier Transform},
  pp. 119--137.

\bibitem{Abe:1994}
S.~Abe and J.~T. Sheridan, ``Generalization of the fractional {F}ourier
  transformation to an arbitrary linear lossless transformation an operator
  approach,'' \emph{J. of Physics A: Math. and General}, vol.~27, no.~12, p.
  4179, 1994.

\bibitem{Pei:2007}
S.~C. Pei and J.~J. Ding, ``Eigenfunctions of {Fourier} and fractional
  {Fourier} transforms with complex offsets and parameters,'' \emph{{IEEE}
  Trans. Circuits Syst. {I}}, vol.~54, no.~7, pp. 1599--1611, Jul. 2007.

\bibitem{Gerrard:1975}
A.~Gerrard and J.~Burch, \emph{Introduction to Matrix Methods in Optics}, ser.
  Dover Books on Physics.\hskip 1em plus 0.5em minus 0.4em\relax Dover, 1975.

\bibitem{Bastiaans:1998}
M.~J. Bastiaans and A.~J. Van~Leest, ``From the rectangular to the quincunx
  {Gabor} lattice via fractional {Fourier} transformation,'' \emph{{IEEE}
  Signal Process. Lett.}, vol.~5, no.~8, pp. 203--205, Aug. 1998.

\bibitem{Zayed:1996}
A.~I. Zayed, ``On the relationship between the {F}ourier and fractional
  {F}ourier transforms,'' \emph{{IEEE} Signal Process. Lett.}, vol.~3, no.~12,
  pp. 310--311, Dec. 1996.

\bibitem{Xiang:2012}
Q.~Xiang and K.~Qin, ``Convolution, correlation, and sampling theorems for the
  offset linear canonical transform,'' \emph{Signal, Image and Video
  Processing}, pp. 1--10, 2012.

\bibitem{Ozaktas:1994}
H.~M. Ozaktas, D.~Mendlovic, L.~Onural, and B.~Barshan, ``Convolution,
  filtering, and multiplexing in fractional {F}ourier domains and their
  relation to chirp and wavelet transforms,'' \emph{Journal of the Optical
  Society of America A}, vol.~11, no.~2, p. 547, Feb. 1994.

\bibitem{Aubel:2017}
C.~Aubel, D.~Stotz, and H.~B\"olcskei, ``A theory of super-resolution from
  short-time {F}ourier transform measurements,'' \emph{Journal of Fourier
  Analysis and Applications}, 2017.

\bibitem{Eldar:2015a}
Y.~C. Eldar, P.~Sidorenko, D.~G. Mixon, S.~Barel, and O.~Cohen, ``Sparse phase
  retrieval from short-time {Fourier} measurements,'' \emph{{IEEE} Signal
  Process. Lett.}, vol.~22, no.~5, pp. 638--642, May 2015.

\bibitem{Zalevsky:1996}
Z.~Zalevsky, R.~G. Dorsch, and D.~Mendlovic, ``{Gerchberg-Saxton} algorithm
  applied in the fractional {F}ourier or the {F}resnel domain,'' \emph{Optics
  Letters}, vol.~21, no.~12, p. 842, Jun. 1996.

\bibitem{Dong:1997}
B.-Z. Dong, Y.~Zhang, B.-Y. Gu, and G.-Z. Yang, ``Numerical investigation of
  phase retrieval in a fractional {F}ourier transform,'' \emph{Journal of the
  Optical Society of America A}, vol.~14, no.~10, p. 2709, Oct. 1997.

\bibitem{Pei:1999}
S.-C. Pei, M.-H. Yeh, and T.-L. Luo, ``Fractional {Fourier} series expansion
  for finite signals and dual extension to discrete-time fractional {Fourier}
  transform,'' \emph{{IEEE} Trans. Sig. Proc.}, vol.~47, no.~10, pp.
  2883--2888, 1999.

\bibitem{Stoica:1997}
R.~L. Stoica, P.~and~Moses, \emph{Introduction to spectral analysis}.\hskip 1em
  plus 0.5em minus 0.4em\relax Prentice hall Upper Saddle River, 1997, vol.~1.

\bibitem{Bhaskar:2013}
B.~N. Bhaskar, G.~Tang, and B.~Recht, ``Atomic norm denoising with applications
  to line spectral estimation,'' \emph{{IEEE} Trans. Sig. Proc.}, vol.~61,
  no.~23, pp. 5987--5999, Dec. 2013.

\bibitem{Bhandari:2015c}
A.~Bhandari, C.~Barsi, and R.~Raskar, ``Blind and reference-free fluorescence
  lifetime estimation via consumer time-of-flight sensors,'' \emph{Optica},
  vol.~2, no.~11, p. 965, Nov. 2015.

\bibitem{Blu:1999}
T.~Blu and M.~Unser, ``Quantitative {Fourier} analysis of approximation
  techniques: {Part I}---{I}nterpolators and projectors,'' \emph{{IEEE} Trans.
  Sig. Proc.}, vol.~47, no.~10, pp. 2783--2795, Oct. 1999.

\bibitem{Matusiak:2012}
E.~Matusiak and Y.~C. Eldar, ``Sub-{Nyquist} sampling of short pulses,''
  \emph{{IEEE} Trans. Sig. Proc.}, vol.~60, no.~3, pp. 1134--1148, Mar. 2012.

\bibitem{Jaganathan:2016}
K.~Jaganathan, Y.~C. Eldar, and B.~Hassibi, ``{STFT} phase retrieval:
  Uniqueness guarantees and recovery algorithms,'' \emph{{IEEE} J. Sel. Topics
  Signal Process.}, vol.~10, no.~4, pp. 770--781, Jun. 2016.

\end{thebibliography}

% Generated by IEEEtran.bst, version: 1.14 (2015/08/26)

\end{document}